\documentclass[journal]{IEEEtran}
\IEEEoverridecommandlockouts
\usepackage{amsmath}
\usepackage[linesnumbered, ruled]{algorithm2e}
\usepackage{amsfonts}
\usepackage{amsthm}
\usepackage{algorithmicx}
\usepackage{array}
\usepackage[caption=false]{subfig}
\usepackage{textcomp}
\usepackage{stfloats}
\usepackage{url}
\usepackage{verbatim}
\usepackage{graphicx}  %插入图片的宏包
\usepackage{float}  %设置图片浮动位置的宏包
\usepackage{subfloat}  %插入多图时用子图显示的宏包

\usepackage{cite}
\usepackage{multicol}

\usepackage{bm}
\usepackage{booktabs}
\usepackage{amssymb}
\usepackage{comment}
\usepackage[colorlinks=true, allcolors=blue]{hyperref}
\usepackage{mathrsfs}
\usepackage{tikz}
\usepackage{pifont}
\newtheorem{definition}{Definition} %中文

\newtheorem{lemma}{Lemma}
\newtheorem{theorem}{Theorem}
\begin{document}

\title{Achievable Sum Rate Optimization on NOMA-aided Cell-Free Massive MIMO with Finite Blocklength Coding
	{\footnotesize \textsuperscript{}}
}

\author{Baolin Chong, Hancheng Lu,~\IEEEmembership{Senior Member,~IEEE}, Yuang Chen, Langtian Qin and Fengqian Guo
\IEEEcompsocitemizethanks{\IEEEcompsocthanksitem Baolin Chong, Hancheng Lu, Yuang Chen, Langtian Qin and Fengqian Guo are with the Department of Electronic Engineering and Information Science, University of Science and Technology of China, Hefei 230027, China. (e-mail: chongbaolin@mail.ustc.edu.cn; hclu@ustc.edu.cn; yuangchen21@mail.ustc.edu.cn; qlt315@mail@mail.ustc.edu.cn; fqguo@ustc.edu.cn) \protect\\}}

\maketitle

\begin{abstract}
Non-orthogonal multiple access (NOMA)-aided cell-free massive multiple-input multiple-output (CFmMIMO) has been considered a promising technology to fulfill the strict quality of service requirements for ultra-reliable low-latency communications (URLLC). However, finite blocklength coding (FBC) in URLLC makes it challenging to achieve optimal performance in the NOMA-aided CFmMIMO system.
In this paper, we investigate the performance of the NOMA-aided CFmMIMO system with FBC in terms of achievable sum rate (ASR). Firstly, we derive a lower bound on the ergodic data rate. Then, we formulate an ASR maximization problem by jointly considering power allocation and user equipment (UE) clustering. To tackle such an intractable problem, we decompose it into two sub-problems, i.e., the power allocation problem and the UE clustering problem.
A successive convex approximation algorithm is proposed to solve the power allocation problem by transforming it into a series of geometric programming problems. Meanwhile, two algorithms based on graph theory are proposed to solve the UE clustering problem by identifying negative loops. Finally, alternative optimization is performed to find the maximum ASR of the NOMA-aided CFmMIMO system with FBC. The simulation results demonstrate that the proposed algorithms significantly outperform the benchmark algorithms in terms of ASR under various scenarios.
\end{abstract}


\section{Introduction}
\IEEEPARstart{U}{ltra}-reliable low-latency communications (URLLC) is an emerging communication technology that has garnered significant research attention due to its vast potential in supporting applications such as virtual reality and augmented reality, haptic internet, vehicle networks, and autonomous systems \cite{Bennis2018,FengUltra2021}. Unlike other communication technologies, URLLC imposes stringent quality of service (QoS) requirements in terms of latency and reliability on the communication process \cite{Nava2020A,Popovski5G2018,HassanKey2021}.
For example, URLLC communication services typically require a target reliability of $1 – 10^{-5}$ within 1 ms user plane latency, according to Third Generation Partnership Project (3GPP) standards \cite{WON2021221}. Hence, URLLC brings huge pressure on wireless communication systems.

Massive multiple-input multiple-output (mMIMO) system has been considered capable of supporting URLLC due to its appealing number of spatial degrees of freedom and channel hardening features. In \cite{9416241}, the authors analyzed the data rate of users in a mMIMO-aided URLLC system with imperfect channel state information (CSI) and pilot contamination.
The design of mMIMO-aided URLLC systems with weighted sum rate and energy efficiency (EE) as the goals are given in \cite{RenJoint2020} and \cite{9516890}, respectively. However, edge users often experience severe path loss in mMIMO systems, which makes it difficult for them to support URLLC. In addition, interference from neighboring cells can cause significant performance degradation for edge users.

To address these issues, cell-free mMIMO (CFmMIMO) systems are employed to effectively support URLLC.
In CFmMIMO systems, numerous distributed access points (APs) jointly serve all user equipments (UEs) through central processing unit (CPU) control on the same time-frequency resource \cite{9586055}.
Unlike mMIMO systems where APs only serve users within their respective cells, CFmMIMO systems eliminate traditional cell boundaries, thereby avoiding inter-cell interference and providing nearly uniform service for UEs.
{\color{blue}
While coordinated multi-point (CoMP) systems can also suppress inter-cell interference, they rely on cooperation between neighboring cells, and interference between non-cooperating cells still exists \cite{electronics12041001, 8761828}. Besides, joint transmission between base stations relies on the interaction of CSI in CoMP systems, limiting its scalable deployment. In CFmMIMO systems, each AP only requires local CSI for precoding \cite{9650567, 8768014}.}
The dense deployment of APs in CFmMIMO systems reduces the distance between APs and UEs, resulting in significant improvements in spectral efficiency (SE) to support URLLC. Furthermore, the stronger channel hardening characteristics of CFmMIMO systems compared to mMIMO systems also ensure the reliability that URLLC is concerned with. However, the interference among UEs in CFmMIMO systems, caused by sharing the same time-frequency resources, can pose a challenge when it comes to supporting URLLC. Non-orthogonal multiple access (NOMA) is applied to CFmMIMO systems to address this issue \cite{Li2018NOMA}. Specifically, UEs are allocated to different clusters, and within the same cluster, UEs decode data through successive interference cancellation (SIC) technology, where transmit power for different users has different levels. Although mitigating interference is attractive, how to cluster UEs and allocate reasonable power to UEs with different channel conditions to ensure successful SIC and optimal performance makes deploying NOMA technology to CFmMIMO systems to better support URLLC challenging.

\subsection{Related Work and Research Gap}
Many existing works focus on characterizing performance within NOMA-aided CFmMIMO systems.
\cite{Li2018NOMA} and \cite{Nguyen2020Max} derived achievable rates for both downlink (DL) and uplink NOMA-aided CFmMIMO systems under Rayleigh fading models, respectively.
The authors of \cite{Zhang2020NOMA} further derived closed-form expressions for the signal-to-interference-plus-noise ratio (SINR) and DL achievable rates under correlated Rayleigh fading models, while \cite{Ohashi2021} considered the case of channel non-reciprocity.
In \cite{Rezaei2020}, achievable rates for NOMA-aided CFmMIMO systems under various precoding techniques, such as maximum ratio transmission (MRT), full-pilot zero-forcing, and modified regularized zero forcing, were derived.
The stochastic geometry approach was also employed in \cite{Kusaladharma2019Achievable} and \cite{Kusaladharma2021Achievable} to more realistically model wireless channel transmission characteristics, focusing on the investigation of achievable rates.  The authors of \cite{Zhang2022Performance} derived the DL SE and EE for the Internet of Things over spatially correlated Rician fading channels.

To improve the performance of NOMA-aided CFmMIMO systems, existing works have carefully designed power allocation algorithms.
In the uplink NOMA-aided CFmMIMO systems, \cite{Nguyen2020Max} and \cite{Zhang2022Performance} proposed power allocation algorithms to maximize the minimum achievable rate and the sum SE, respectively.
Two algorithms for allocating the power in DL NOMA-aided CFmMIMO systems were proposed in \cite{Zhang2019Spectral} and \cite{Le2021Learning}, aiming to maximize the minimum SE and the sum SE, respectively.
In addition to power allocation, the UE clustering scheme plays an important role in the NOMA-aided CFmMIMO system.
The authors of \cite{BasharOn2020} initially studied direct UE pairing schemes, including random pairing, close pairing, and far pairing.
The authors of \cite{Dang2022Optimal} further studied UE pairing schemes by relaxing binary variables into continuous variables and using an inner approximation method to solve the problem of maximizing the minimum achievable rate.
In \cite{Le2021Learning}, the K-means++ and improved K-means++ algorithms were introduced for UE clustering in NOMA-aided CFmMIMO systems.

Existing works have proved that NOMA-aided CFmMIMO systems can significantly enhance the performance of wireless communication systems.
Therefore, we aim to maximize the achievable sum rate (ASR) by jointly optimizing power allocation and UE clustering in the NOMA-aided CFmMIMO system to support URLLC.
However, for URLLC, the performance analysis and optimization on NOMA-aided CFmMIMO systems should be reconsidered.
Firstly, there is a lack of analysis on the ergodic rates of URLLC UEs in the NOMA-aided CFmMIMO system, which needs to be derived.
Besides, short packets are used for data transmission in URLLC to ensure low transmission latency and simplify the decoding complexity at the receiver \cite{Polyanskiy2010,ZhaoQueue2022}. In this case, the Shannon capacity, which is based on the law of large numbers, is no longer applicable due to the non-negligible decoding error probability caused by finite blocklength coding (FBC). In \cite{Polyanskiy2010}, an approximate expression for the maximum achievable rate with respect to decoding error probability, code length, and signal-to-noise ratio (SNR) under FBC has been derived. Unlike the Shannon formula, the maximum achievable rate under FBC is neither convex nor concave with respect to SNR \cite{RenJoint2020}.
Therefore, the power allocation algorithm proposed in \cite{Zhang2019Spectral, Le2021Learning} and the UE clustering method presented in \cite{Dang2022Optimal} are no longer applicable.
The clustering scheme proposed in \cite{Le2021Learning}, which relies solely on large-scale fading, also cannot guarantee the maximization of ASR.

\subsection{Contribution and Outline}
In this paper, we theoretically analyze the impact of FBC on the NOMA-aided CFmMIMO system and attempt to achieve optimal performance in terms of ASR by jointly optimizing power allocation and UE clustering. To the best of our knowledge, this is the first study on the NOMA-aided CFmMIMO system for URLLC. The main contributions of this paper are summarized as follows:

\begin{itemize}
	
	\item We propose a NOMA-aided CFmMIMO system with consideration of URLLC in the finite blocklength regime, where all APs serve URLLC UEs in different clusters simultaneously. We analyze the system performance and derive the lower bound (LB) for the ergodic data rate of URLLC UEs under FBC. To optimize the system performance, an ASR maximization problem is formulated by jointly optimizing power allocation and UE clustering.
To tackle such an intractable problem, we decompose it into two sub-problems, i.e., the power allocation problem and the UE clustering matrix design problem. Then the original problem can be efficiently solved by a two-step iterative optimization algorithm.
	
	\item We use successive convex approximation (SCA) to solve the power allocation problem. Specifically, the objective function (OF) is first transformed into a convex function using logarithmic transformation. Then, the complex constraint conditions generated during the transformation process are converted into convex constraints through scaling. By utilizing SCA, the original power allocation problem is transformed into a series of geometric programming (GP) problems that can be solved efficiently.

	\item We solve the UE clustering matrix design problem based on graph theory. Specifically, we first reformulate the UE clustering problem as a negative loop detection problem in a weight directed graph. Then, we construct a weight directed graph based on the current clustering situation. Finally, we adopt two negative loop searching algorithms to effectively find the negative loops in the graph.
\end{itemize}

Simulation results validate the tightness of the derived LB and demonstrate that the proposed algorithms outperform the benchmark algorithms in terms of ASR under various scenarios.

The rest of the paper is organized as follows. Section II gives the DL NOMA-aided CFmMIMO system model and formulates the joint power allocation and UE clustering problem for maximizing the ASR.
An iterative optimization algorithm is proposed in Section III, where the power allocation problem is transformed into a series of GP problems based on SCA and the UE clustering matrix design problem is converted into the problem of identifying \emph{differ-cluster negative loop} based on graph theory. Section IV presents the numerical results and analysis. Finally, the conclusion is drawn in Section V.

\emph{Notations:} In this paper, vectors and matrices are denoted by lowercase and uppercase bold letters, respectively. $x_{i,j}$ represents the $i$-th row and $j$-th column element of matrix $\mathbf{X}$. $\mathbf{X}^H$ represents the Hermitian of matrix $\mathbf{X}$. $\pi_i$ denotes the $i$-th elements of vector $\pmb{\pi}$. $\mathbb{C} ^{i\times j}$ represents space of $i\times j$ complex number matrices. $\mathbb{E}[i]$ denotes the expected value of $i$.

\section{System Model and Problem Formation}

\begin{figure}[ht]
	\centering
	\includegraphics[scale=0.6]{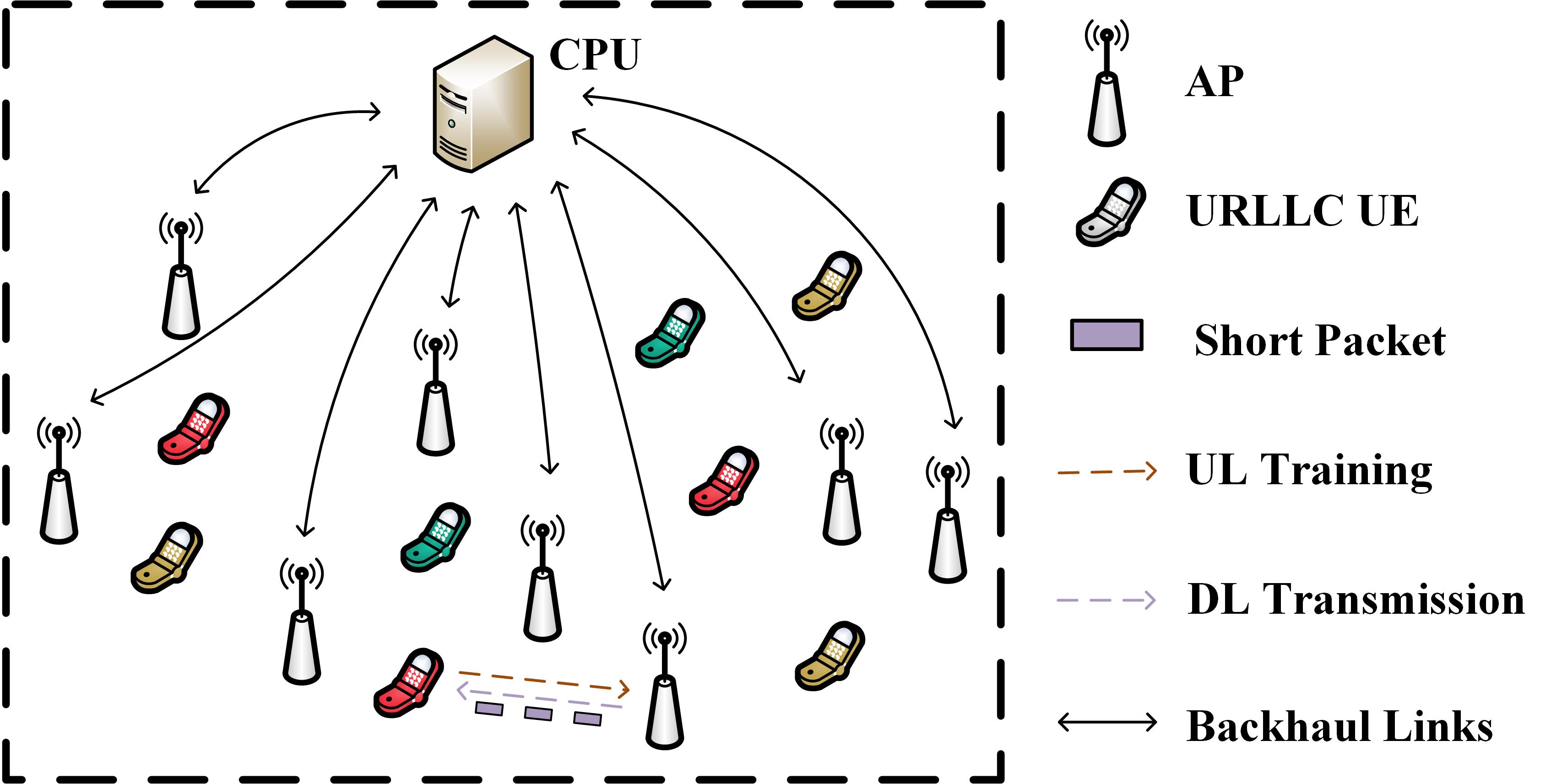}
	\caption{The NOMA-aided CFmMIMO system for URLLC.}
	\label{fig:model}
\end{figure}

We consider a DL NOMA-aided CFmMIMO system, where $N$ single antenna URLLC UEs are served by $M$ APs that are equipped with $L$ antennas each, as illustrated in Fig. \ref{fig:model}.
The URLLC UEs are organized into $G$ clusters, and each cluster is represented by a different color in the Fig \ref{fig:model}.
For convenience, we define $\mathcal{M} \triangleq \left \{ 1,2,\cdots,M \right \}$, $\mathcal{N} \triangleq \left \{ 1,2,\cdots,N \right \}$, and $\mathcal{G} \triangleq \left \{ 1,2,\cdots,G \right \}$ as the index of APs, URLLC UEs, and clusters.
Communication between APs and URLLC UEs follows the TDD protocol, where each coherent interval $\tau_c = \tau_p + \tau_d$ is divided into uplink training $\tau_p$ and downlink data transmission $\tau_d$.
To represent the association between URLLC UEs and clusters, we employ a matrix $\mathbf{X} = \left [ x_{gn} \right ]_{g \in \mathcal{G},n \in\mathcal{N} }\in \mathbb{C}^{G\times N}$ and a vector $\pmb{\pi} = \left [ \pi_{n} \right ]_{n \in\mathcal{N} } \in \mathbb{C}^{N\times 1}$. Specifically, if URLLC UE $n$ is allocated to cluster $g$, then $x_{gn} = 1$ and $\pi_n = g$; otherwise $x_{gn} = 0$.
We consider the block fading model \cite{RenJoint2020, Peng2023, PengResourceJul,10138403}, and the channel between AP $m$ and URLLC UE $n$ is $\mathbf{h}_{mn} = \sqrt{\beta_{mn} } \pmb{\zeta }_{mn}$, where $\beta_{mn}$ represents the large-scale fading coefficient influenced by path loss and shadow fading, and $\pmb{\zeta }_{mn} \in \mathbb{C}^{L \times 1}$ denotes the small scale fading. Each element of $\pmb{\zeta}_{mn}$ is independent complex Gaussian random variables with zero mean and unit variance, i.e., $\mathcal{CN}(0,1)$.
Define the weighted directed graph of the system, where all URLLC UEs form the nodes, and the relationships between URLLC UEs belonging to different clusters form the edges in the graph.
%Table \ref{t1} lists a compilation of the main notations used in this paper.

\subsection{Uplink Pilot Training and Channel Estimation}
In the uplink pilot training phase, all UEs need to transmit pilot sequences for channel estimation.
Within the same cluster, URLLC UEs use a shared pilot sequence for channel estimation, while URLLC UEs in different clusters send orthogonal pilot sequences \cite{Rezaei2020}. Therefore, the length of the pilot sequence is equal to the number of clusters, i.e., $\tau_p = G$.
Specifically, the pilot sequence transmitted by URLLC UEs in cluster $g$ is denoted by $\pmb{\phi }_{g} \in \mathbb{C}^{G\times 1} $, where $\left \| \pmb{\phi }_{g} \right \| ^2 = 1$.
Besides, it holds that $\pmb{\phi }_{g}^H \pmb{\phi }_{g'}=0, \forall g\ne g'$ for different clusters $g$ and $g'$ due to orthogonality.
The received sequence at AP $m$ is given by
\begin{equation}\label{received signal at AP m}
  \begin{aligned}
    \mathbf{Y}_{m}^p = \sqrt{G p_p}\sum_{n\in \mathcal{N}}\sum_{g\in \mathcal{G}}x_{gn}\mathbf{h}_{mn}\pmb{\phi}_g^H + \mathbf{N}_m,
  \end{aligned}
\end{equation}
where $p_p$ denotes the normalized pilot power, and $\mathbf{N}_m \in \mathbb{C}^{L\times G}$ is the Gaussian noise matrix with i.i.d $\mathcal{CN}(0,1)$ elements at AP $m$.
After receiving the sequences, all APs estimate instantaneous CSI with all URLLC UEs using minimum mean square error (MMSE) estimation.
Since the received signal $\mathbf{y}_{m\pi_n}^p$ follows a Gaussian distribution \cite{Rezaei2020}, the MMSE estimate $\mathbf{\hat{h}_{mn}}$ can be reexpressed as
\begin{equation}\label{theta}
  \mathbf{\hat{h}}_{mn} = \sqrt{\theta_{mn}}\pmb{\nu}_{m\pi_n} ,
\end{equation}
where $\theta_{mn} = \frac{G p_p\beta_{mn}^2 }{1 + G p_p \sum_{n'=1}^{N}x_{\pi_nn'}\beta_{mn'} } $ and $\pmb{\nu}_{m\pi_n} $ represents a circularly symmetric complex Gaussian random vector with zero mean and identity covariance matrix.

\subsection{DL Data Transmission}

{\color{blue}
During the DL data transmission phase, all APs first perform superposition coding for each clustering. The superposition coded data signal for the $g$-th cluster in the AP $m$ is given by
\begin{equation}\label{superposition coded}
  \begin{aligned}
     \bar{q}_{mg} =\sum_{n\in \mathcal{N}} x_{gn} \sqrt{p_{mn}} q_n,
  \end{aligned}
\end{equation}
where $q_n$ and $p_{mn}$ denote the data signals transmitted to URLLC UE $n$ and the power allocated to URLLC UE $n$ by AP $m$, respectively.
Then, each AP uses distributed MRT for precoding, where the precoding vector between AP $m$ and URLLC UE $n$ is $\pmb{\nu}_{m\pi_n}$ \cite{Rezaei2020, 9586055}.
The signal transmitted by AP $m$ is given by
\begin{equation}\label{transmitted signal by AP m}
  \begin{aligned}
    \mathbf{w}_m = \sum_{g\in \mathcal{G}} \pmb{\nu}_{mg}\bar{q}_{mg}= \sum_{n\in \mathcal{N}}\sqrt{p_{mn}} \pmb{\nu}_{m\pi_n} q_n.
  \end{aligned}
\end{equation}}
All URLLC UEs receive signals from all APs, and the received signal at URLLC UE $n$ is expressed as
\begin{equation}\label{received signal at UE}
  \begin{aligned}
    y_n = \sum_{m=1}^{M} \sum_{n'=1}^{N} \sqrt{p_{mn'}}\mathbf{h}_{mn}^{H}\pmb{\nu}_{m\pi_{n'}}q_{n'} + n_n,
  \end{aligned}
\end{equation}
where $n_n \sim \mathcal{CN}(0,1)$ represents the noise at URLLC UE $n$.

In the NOMA-aided CFmMIMO system, due to channel reciprocity, it can be assumed that the uplink and downlink CSI remains unchanged. To implement power-domain NOMA, we arrange URLLC UEs in descending order based on the mean of the effective channel gains as follows \cite{Rezaei2020}
\begin{equation}\label{rank condition}
  \Omega_1 \ge \Omega_2 \ge \dots \ge \Omega_N,
\end{equation}
where $\Omega_n = \mathbb{E}\left  \{ \left | {\textstyle \sum_{m=1}^{M}}  \pmb{\nu}_{m\pi_n}^H\hat{\mathbf{h}}_{mn} \right |^2 \right \}$, $\forall n\in\mathcal{N}$.
Then, SIC is employed and URLLC UEs within the same cluster are decoded sequentially according to (\ref{rank condition}).
{\color{blue}
Following the principles of NOMA, URLLC UE $n$ initially decodes URLLC UE $n_1 > n$ with poorer channel condition within the same cluster. Subsequently, its own data is successively decoded after removing the interference from these URLLC UEs \cite{ZhangSWIPT2023,Rezaei2020}.}
To achieve effective deployment of SIC, the following conditions need to be met \cite{ZhangSWIPT2023,Rezaei2020}.
\begin{equation}\label{conditions for SIC}
  \begin{aligned}
      \mathbb{E}\left \{ R_n(\gamma_n^{n_1}) \right \} \ge \mathbb{E}\left \{ R_{n_1}(\gamma_{n_1}^{n_1}) \right \}, n \le n_1 \cap \pi_n = \pi_{n_1},
  \end{aligned}
\end{equation}
where $R_n(\gamma_n^{n_1})$ represents the achievable rate of URLLC $n$ with SINR $\gamma_n^{n_1}$, $\gamma_n^{n_1}$ and $\gamma_{n_1}^{n_1}$ denote the SINR in decoding the signal of URLLC UE $n_1$ by URLLC UE $n$ and itself, respectively.
According to (\ref{conditions for SIC}), the SINR of URLLC UE $n$ is defined as $\gamma_n = \min \left ( \gamma_n^{n},\gamma_{n_1}^n \right )$, $\forall n_1 \le n \cap \pi_{n_1} = \pi_n$, to guarantee that URLLC UE $n_1$ can perform SIC and decode the data of URLLC UE $n$.

However, achieving a perfect SIC is infeasible due to statistical CSI knowledge limitations, channel estimation error, and pilot contamination within the cluster \cite{Rezaei2020}. Hence, the received signal after an imperfect SIC process is given in (\ref{receive signal after SIC}), where $\iota_{mnn'} = \mathbf{h}_{mn}^{H}\pmb{\nu}_{m\pi_{n'}}$, $Y_{ds,n}$, $Y_{bu,n}$, $Y_{ui,n}$, $Y_{ici,n}$, and $Y_{rici,n}$ represent the desire signal, uncertainty of precoding gain, intra-cluster interference, intra-cluster interference after SIC and residual interference due to imperfect SIC for URLLC UE $n$, respectively, and $\hat{q}_n$ represents the estimate of $q_{n}$.
The correlation between the estimated parameter $\hat{q}_n$ and its actual value $q_n$ can be modeled as
\begin{equation}\label{relation between q and estimation q}
  \begin{aligned}
    q_{n} = c_{n}\hat{q}_{n}+e_{n},
  \end{aligned}
\end{equation}
where $\hat{q}_n \sim \mathcal{CN}(0,1)$, $ e_n\sim \mathcal{CN}(0,\sigma_{e_n}^2/\left ( 1+\sigma_{e_n}^2 \right ) )$, and $ c_n = 1/\sqrt{1+\sigma_{e_n}^2} $. Moreover, $\hat{q}_n$ and $e_n$ are statistically independent.
Then, the expression of $\gamma_n^n$ is given by
\begin{equation}\label{SINR of MRT}
  \begin{aligned}
    \gamma_n^n = \frac{ \left | Y_{ds,n} \right |^2 }{ \left | Y_{ici,n} \right |^2 +\left | Y_{rici,n} \right |^2 +\left | Y_{bu,n} \right |^2 +\left | Y_{ui,n} \right |^2  +n_n^2}.
  \end{aligned}
\end{equation}

\begin{figure*}
   \begin{equation}\label{receive signal after SIC}
     \begin{aligned}
     \tilde{y}_n = &\underbrace{\sum_{m=1}^{M} \sqrt{p_{mn}}  \mathbb{E}\left \{\iota_{mnn}\right \}}_{Y_{ds,n}}q_{n}+\underbrace{\sum_{m=1}^{M}  \sqrt{p_{mn}}\left (\iota_{mnn} -\mathbb{E}\left \{\iota_{mnn}  \right \} \right )q_{n}}_{Y_{bu,n}}+\underbrace{\sum_{m=1}^{M} \sum_{ g'\ne \pi_n}^{G}  \sum_{n'=1}^{N} x_{g'n'} \sqrt{p_{mn'}}\iota_{mnn'}q_{n'}}_{Y_{ui,n}}\\
                 & + \underbrace{\sum_{m=1}^{M} \sum_{n'=0}^{n-1} x_{\pi_nn'} \sqrt{p_{mn'}}\iota_{mnn'}q_{n'}}_{Y_{ici,n}}+ \underbrace{\sum_{m=1}^{M} \sum_{n'=n+1}^{N} x_{\pi_nn'} \sqrt{p_{mn'}} \left (\iota_{mnn'}q_{n'} - \mathbb{E}\left \{ \iota_{mnn'} \right \}\hat{q}_{n'} \right )}_{Y_{rici,n}} + n_n .
     \end{aligned}
   \end{equation}
  {\noindent} \rule[-10pt]{18cm}{0.05em}
\end{figure*}

\subsection{Achievable Data Rate for URLLC UE}
To meet the stringent QoS requirements of URLLC, FBC, which can reduce transmission latency and decrease decoding complexity, is employed.
However, short packet transmissions lead to an increase in decoding error probability, which cannot be ignored.
Hence, the assumption of error-free Shannon capacity with infinite blocklength is no longer applicable.
Denote $\epsilon_n$ as the maximum decoding error probability of URLLC UE $n$. The data rate under FBC for URLLC UE $n$ can be approximated as \cite{Polyanskiy2010}
\begin{equation} \label{rate of user n}
  R_n(\gamma_n) = \eta \log_2(1+ \gamma_n) - \sqrt{\frac{\eta V(\gamma _n)}{\tau_d} } \frac{Q^{-1}{(\epsilon _n)}}{\ln 2} ,
\end{equation}
where $\eta = \tau_d/\tau_c$, $ Q^{-1}{(\epsilon _n)}$ represents the inverse of the Gaussian Q-function, and $V(\gamma _n)= 1 - (1 + \gamma _n)^{-2}$ denotes the channel dispersion. The ergodic data rate for URLLC UE $n$ under FBC is $\bar{R}_n = \mathbb{E}_{\gamma_n} \left \{ \max(R_n(\gamma_n),0) \right \}$.
Obtaining a closed-form expression for $\bar{R}_n$ and optimizing it presents a significant challenge. To address this challenge, we can derive the LB of $\bar{R}_n$ that captures the ergodic data rate and facilitates optimization.
Define $r_n(\frac{1}{\gamma_n}) = R_n(\gamma_n)$, where $r_n(x) = \log_2(1+ \frac{1}{x}) - \frac{Q^{-1}{(\epsilon _n)}}{\ln2\sqrt{n_d}} \sqrt{\frac{1+2x}{(1+x)^2} }$. To ensure $R_n(x)\ge0$, we require inequality $f\left(\frac{1}{x}\right) \triangleq \frac{(1+\frac{1}{x})\log_2(1+x)}{\sqrt{1+\frac{2}{x}}} \ge \frac{Q^{-1}(\epsilon_n)}{\ln2\sqrt{n_d}}$ to hold. It is obvious that the first derivative of $f(x)$ is negative, indicating the monotonic decrease of $f(x)$. Thus, the domain that makes $r(x)>0$ is $0 < x < f^{-1}\left ( \frac{Q^{-1}(\epsilon_n)}{\ln2\sqrt{n_d}} \right )$ \cite{Peng2023}.
Since the function $f(x)$ is decreasing and convex within this domain, the following conclusions can be drawn by utilizing Jensen's inequality:
\begin{equation}\label{LB of ergodic rate}
  \hat{R}_n \triangleq R \left ( \bar{\gamma}_n \right ) \le\mathbb{E}_{\gamma_n}\left \{ R(\gamma_n) \right \}  \le   \bar{R}_n,
\end{equation}
where $\hat{R}_n$ is the LB of the ergodic data rate and $\bar{\gamma}_n = \mathbb{E}_{\gamma_n}^{-1}\left \{ \gamma_n^{-1} \right \}$. The expression of $\hat{R}_n$ is derived in the following theorem.
\begin{theorem}\label{theorem of LB}
  The LB of the ergodic data rate for URLLC UE $n$ ,$\forall n \in \mathcal{N}$ under FBC in the NOMA-aided CFmMIMO system can be given by\footnote{ The LB of the ergodic rate is affected by the accuracy of channel estimation. Too small a value of $\theta_{mn}$ directly leads to a large $\bar{\gamma}_{n}$, and the corresponding LB of the ergodic rate can be very low. In addition, the LB of the ergodic rate exhibits an upper bound as the number of antennas increases. As the number of antennas per AP increases, $\bar{\gamma}_{n}$ gradually decreases and stabilizes, resulting in the LB of the ergodic rate becoming bounded.}
  \begin{equation}\label{LB of rate}
    \hat{R}_n= \eta \log_2(1+ \bar{\gamma}_n) - \sqrt{\frac{\eta V(\bar{\gamma} _n)}{\tau_d} } \frac{Q^{-1}{(\epsilon _n)}}{\ln 2},
  \end{equation}
  where $\bar{\gamma}_n = \min(\bar{\gamma}_n^n,\bar{\gamma}_{n_1}^n), \forall n_1 \le n \cap \pi_n = \pi_{n_1}$, $\bar{\gamma}_n^n$ and $\bar{\gamma}_{n_1}^n$ are given by (\ref{bar gamma}) and (\ref{bar gamma2}), respectively.
\end{theorem}
\begin{proof}
  Please refer to Appendix \ref{appendix A}.
\end{proof}

\begin{figure*}
  \begin{equation}\label{bar gamma}
    \begin{aligned}
      \bar{\gamma}_n^n = \frac{L\left (\sum\limits_{m=1}^{M} \sqrt{p_{mn}\theta_{mn}}  \right )^2  }{ \sum\limits_{n'=1}^{N}\sum\limits_{m=1}^{M}p_{mn'}\beta_{mn} + L\sum\limits_{n'=1}^{n-1}x_{\pi_nn'}\left (\sum\limits_{m=1}^{M} \sqrt{p_{mn'}\theta_{mn}}  \right )^2 +L\sum\limits_{n'=n+1}^{N}\left ( 2-2c_{n'} \right ) x_{\pi_nn'}\left (\sum\limits_{m=1}^{M} \sqrt{p_{mn'}\theta_{mn}}  \right )^2 +1 }.
    \end{aligned}
  \end{equation}
  \begin{equation}\label{bar gamma2}
    \begin{aligned}
      \bar{\gamma}_{n_1}^n = \frac{L\left (\sum\limits_{m=1}^{M} \sqrt{p_{mn}\theta_{mn_1}}  \right )^2  }{ \sum\limits_{n'=1}^{N}\sum\limits_{m=1}^{M}p_{mn'}\beta_{mn_1} + L\sum\limits_{n'=1}^{n-1}x_{\pi_nn'}\left (\sum\limits_{m=1}^{M} \sqrt{p_{mn'}\theta_{mn_1}}  \right )^2 +L\sum\limits_{n'=n+1}^{N}\left ( 2-2c_{n'} \right ) x_{\pi_nn'}\left (\sum\limits_{m=1}^{M} \sqrt{p_{mn'}\theta_{mn_1}}  \right )^2 +1 }.
    \end{aligned}
  \end{equation}
  {\noindent} \rule[-10pt]{18cm}{0.05em}
\end{figure*}

\subsection{Problem Formation}

URLLC cares about latency and reliability and requires low latency and high reliability. Transmitting a substantial amount of data within a given unit of time ensures the rapid transfer of information while maintaining reliability. Therefore, achievable data rates are a key metric that can effectively gauge the performance of URLLC in terms of latency and reliability. The achievable rate is used as an indicator to measure the ability to support URLLC in mMIMO \cite{RenJoint2020} and CFmMIMO systems\cite{Peng2023, PengResourceJul}. In addition, due to the impact of channel hardening in NOMA-aided CFmMIMO systems, random channel gains can be negligible. Therefore, we aim to jointly optimize power allocation and UE clustering based on large-scale fading information to maximize ASR.
{\color{blue}
Note that optimization is based on large-scale fading information rather than small-scale fading, which is beneficial for URLLC \cite{RenJoint2020}. In other words, the algorithm only needs to be rerun when the large-scale fading information, varying slowly compared to the small-scale fading, has changed. \footnote{{\color{blue}The precoding at the APs and the proposed algorithms are two separate components of the system. Precoding needs to be performed during the downlink transmission phase at each coherent interval, but the proposed algorithms do not require corresponding rerun.}}}
Mathematically, the optimization problem can be formulated as
\begin{subequations}\label{orignal optimization problem}
\begin{align}
\max_{\mathbf{P}, \mathbf{X}  } \sum_{n \in \mathcal{N}}\hat{R}_n   \ \ \tag{\ref{orignal optimization problem}}
\end{align}
\begin{alignat}{2}
\text{s.t.}\ & \hat{R}_n \ge \hat{R}_n^{req},\  \forall n \in \mathcal{N},   \label{orignal constrited 1} \\
    & \sum_{n \in \mathcal{N}} p_{mn} \le p_{max},  \forall m \in \mathcal{M},  \label{orignal constrited 2} \\
    &p_{mn} \le p_{mn'},\  \pi_n=\pi_{n'}, n\le n', \forall n,n' \in \mathcal{N}, \forall m \in \mathcal{M},    \label{orignal constrited 4} \\
    &\sum_{g \in \mathcal{G}}x_{gn} = 1,\ x_{gn}\in\left \{ 0,1 \right \},\ \forall n \in \mathcal{N}, \label{orignal constrited 3}
\end{alignat}
\end{subequations}
where $\mathbf{P} = \left [ p_{mn} \right ]_{m\in\mathcal{M},n\in \mathcal{N}}$, $\hat{R}_n^{req}$ and $p_{max}$ represent the minimum rate requirement for URLLC UE $n$ and maximum transmit power at each AP, constraint (\ref{orignal constrited 1}) ensures that each URLLC UE meets its minimum data rate requirement, constraint (\ref{orignal constrited 2}) limits the maximum transmission power at each AP, constraint (\ref{orignal constrited 4}) is the necessary condition to implement SIC, and constraint (\ref{orignal constrited 3}) ensures that each URLLC UE is allocated to one cluster and is exclusively assigned to that cluster.

\section{Problem Analysis and Solution}
In order to solve mixed-integer non-linear programming problem (\ref{orignal optimization problem}) with affordable complexity, we propose a tractable algorithm that utilizes alternating optimization to separately and iteratively solve $\mathbf{P}$ and $\mathbf{X}$. Specifically, our algorithm involves decomposing the original problem, which aims to maximize the ASR, into two subproblems: DL power allocation with fixed UE clustering and UE clustering matrix design with fixed power allocation.
Through the alternating optimization process, the global ASR continues to increase and eventually converges, since the OF is upper-bounded within the feasible set. In the rest of this section, the details of the algorithm are described.

\subsection{DL power allocation based on SCA}
With the fixed UE clustering matrix, the optimization problem for DL power allocation can be expressed as
\begin{subequations}\label{optimization problem 11}
\begin{align}
\max_{\mathbf{P}} \sum_{n\in \mathcal{N}} \frac{\eta}{\ln 2}\left ( \ln(1+\bar{\gamma}_n)-a_nM(\bar{\gamma}_n) \right )    \ \ \tag{\ref{optimization problem 11}}
\end{align}
\begin{alignat}{2}
\text{s.t.}\  (\ref{orignal constrited 1}), (\ref{orignal constrited 2}), (\ref{orignal constrited 4}),
\end{alignat}
\end{subequations}
where $a_n = \frac{ Q^{-1}(\epsilon _n)}{\sqrt{\eta \tau_d}}$ and $M(\bar{\gamma}_n) = \sqrt{1-(1+\bar{\gamma}_n^{-1})^{-2}}$. To simply the problem (\ref{optimization problem 1}) as a GP, the auxiliary variables $\kappa_n$, $\forall n \in \mathcal{N}$, is introduced and then problem (\ref{optimization problem 1}) can be reexpressed as
 \begin{subequations}\label{optimization problem 1}
\begin{align}
\max_{\mathbf{P},\pmb{\kappa}} \sum_{n\in \mathcal{N}}\frac{\eta}{\ln 2}\left ( \ln(1 + \kappa_n)-a_nM(\kappa_n) \right )    \ \ \tag{\ref{optimization problem 1}}
\end{align}
\begin{alignat}{2}
\text{s.t.}\ & \kappa_n \le \bar{\gamma}_n^n ,\ \forall n \in \mathcal{N}, \label{problem 11 constrited 1} \\
    & \kappa_n \le \bar{\gamma}_{n_1}^n ,\ \forall n_1\le n \cap \pi_n = \pi_{n_1},\forall n,n_1 \in \mathcal{N}, \label{problem 11 constrited 2} \\
    &  \kappa_n \ge \hat{\gamma}_n^{req},\ \forall n \in \mathcal{N}, \label{problem 11 constrited 3} \\
    &(\ref{orignal constrited 2}), (\ref{orignal constrited 4}),
\end{alignat}
\end{subequations}
where $\pmb{\kappa} = \left [ \kappa_1, \kappa_2,\cdots,\kappa_N \right ]$, and $\hat{\gamma}_n^{req} = {R}_n^{-1} \left ( \hat{R}_n^{req}  \right ) $ represents the minimum SINR requirement for URLLC UE $n$. The OF of the problem (\ref{optimization problem 1}) is obviously non-convex.
To address this problem, we given the lower bound $\ln(1 + \kappa_n) \ge \rho_n \ln (\kappa_n) + \chi_n,$ and the upper bound $M(\kappa_n) \le \hat{\rho}_n \ln (\kappa_n)+ \hat{\chi}_n$ with given $\bar{\kappa}_n$ according to Lemma 3 and Lemma 4 of \cite{10041787}, respectively, where $\rho_n$, $\chi_n$, $\hat{\rho}_n$, and $\hat{\chi}_n$ are given by
\begin{equation}\label{expression of 1}
  \begin{aligned}
    \rho_n = \frac{\bar{\kappa}_n}{1 + \bar{\kappa}_n},  \chi_n = \ln\left ( 1 + \bar{\kappa}_n\right ) -  \frac{\bar{\kappa}_n}{1 + \bar{\kappa}_n} \ln\left ( \bar{\kappa}_n\right ),
  \end{aligned}
\end{equation}
\begin{equation}\label{expression of 2}
  \begin{aligned}
    \hat{\rho}_n = \frac{\bar{\kappa}_n}{\sqrt{\bar{\kappa}_n^2+2\bar{\kappa}_n}}-\frac{\bar{\kappa}_n\sqrt{\bar{\kappa}_n^2+2\bar{\kappa}_n}}{\left ( 1+\bar{\kappa}_n \right )^2 },
  \end{aligned}
\end{equation}
\begin{equation}\label{expression of 3}
  \begin{aligned}
    \hat{\chi}_n = \sqrt{1-{(1+\bar{\kappa}_n)^{-2}} }-\hat{\rho}_n ln(\bar{\kappa}_n).
  \end{aligned}
\end{equation}

Then, the OB of the problem (\ref{optimization problem 1}) can be simplified. However, due to the complexity of constraints (\ref{problem 11 constrited 1}) and (\ref{problem 11 constrited 2}), problem (\ref{optimization problem 1}) still cannot be directly solved. Thus, we handle constraints (\ref{problem 11 constrited 1}) and (\ref{problem 11 constrited 2}) to transform problem (\ref{optimization problem 1}) into a GP problem.
Denote the numerators of $\bar{\gamma}_{n}^n$ and $\bar{\gamma}_{n_1}^n$ as $\zeta_{n}^2$ and $\hat{\zeta}_{n_1n}^2$, respectively, with denominators represented by $\varpi_n$ and $\hat{\varpi}_{n_1n}$.
Based on the Theorem 4 of \cite{10041787}, We can obtain the lower bound of $\zeta_{n}$ in the form of a monomial function with given $\bar{p}_{mn}$ as follows
\begin{equation}\label{reexpression of zetan}
  \begin{aligned}
    \zeta_n = \sum\limits_{m=1}^{M} \sqrt{L p_{mn}\theta_{mn}} \ge c_n \prod_{m\in \mathcal{K}}\left ( L p_{mn}\theta_{mn} \right )^{a_{mn}},
  \end{aligned}
\end{equation}
where $c_n$ and $a_{mn}$ are given by
\begin{equation}\label{coef 111}
  \begin{aligned}
    &a_{mn} = \sqrt{L \bar{p}_{mn}\theta_{mn}}/ (2 \zeta_n^*),c_n = \zeta_n^* \prod_{m\in \mathcal{K}}\left ( L \bar{p}_{mn}\theta_{mn} \right )^{-a_{mn}},
  \end{aligned}
\end{equation}
and $\zeta_n^*$ is obtained by using $ p_{mn} = \bar{p}_{mn} $.
Similarly, the lower bound of $\hat{\zeta}_{n_1n}$ in the form of a monomial function with given $\bar{p}_{mn}$ as follows
\begin{equation}\label{reexpression of zetan1n}
  \begin{aligned}
    \hat{\zeta}_{n_1n} = \sum\limits_{m=1}^{M} \sqrt{L p_{mn}\theta_{mn_1}} \ge \hat{c}_{n_1n} \prod_{m\in \mathcal{K}}\left ( L p_{mn}\theta_{mn_1} \right )^{\hat{a}_{mn_1n}},
  \end{aligned}
\end{equation}
where $c_{n_1n}$ and $a_{mn_1n}$ are given by
\begin{equation}\label{coef 222}
  \begin{aligned}
    &a_{mn_1n} = \sqrt{L \bar{p}_{mn}\theta_{mn_1}}/ (2 \hat{\zeta}_{n_1n}^*),\\
    &c_{n_1n} = \hat{\zeta}_{n_1n}^* \prod_{m\in \mathcal{K}}\left ( L \bar{p}_{mn}\theta_{mn_1} \right )^{-a_{mn_1n}},
  \end{aligned}
\end{equation}
and $\hat{\zeta}_{n_1n}^*$ is obtained by using $p_{mn} = \bar{p}_{mn}$.

Based on the analysis above, we can approximate the OF and constraints of the problem (\ref{optimization problem 1}) and then solve the approximate problem in an iterative manner. In the following, we provide a detailed explanation of the iterative process.

Firstly, in the $i$-th iteration, we denote the power allocation coefficient and the auxiliary variable as $\mathbf{P}^{\left ( i \right ) }$ and $\pmb{\kappa}^{\left ( i \right )}$, respectively.
Correspondingly, $\rho_n^{\left ( i \right ) }$, $\hat{\rho}_n^{\left ( i \right ) }$ can be obtained based on (\ref{expression of 1}) and (\ref{expression of 2}) for simplify OB. The lower bound for $n$-th term of OB can be expressed
\begin{equation}\label{nth term lower bound}
  \begin{aligned}
    &\frac{\eta}{\ln 2}\left ( \ln(1 + \kappa_n)-a_nM(\kappa_n) \right ) \\
    &\ge \frac{\eta}{\ln 2}\left ( \left ( \rho_n^{\left ( i \right ) } -a_n \hat{\rho}_n^{\left ( i \right ) } \right ) \ln (\kappa_n) + \chi_n^{\left ( i \right ) }  -a_n \hat{\chi}_n^{\left ( i \right ) } \right ),
  \end{aligned}
\end{equation}
where $\chi_n^{\left ( i \right ) }$ and $\hat{\chi}_n^{\left ( i \right ) }$ do not need to be calculated as the constant terms in the OB are omitted.
Simultaneously, $a_{mn}^{\left ( i \right ) }$, $c_n^{\left ( i \right ) }$, $a_{mn_1n}^{\left ( i \right ) }$, and $c_{n_1n}^{\left ( i \right ) }$ can be computed based on (\ref{coef 111}) and (\ref{coef 222}) to simplify constraints (\ref{problem 11 constrited 1}) and (\ref{problem 11 constrited 2}). The upper bound for $\zeta_n$ and $\hat{\zeta}_{n_1n}$ can be expressed as
\begin{equation}\label{2 reexpression of zetan}
  \begin{aligned}
    \zeta_n \ge c_n^{\left ( i \right ) } \prod_{m\in \mathcal{K}}\left ( L p_{mn}\theta_{mn} \right )^{a_{mn}^{\left ( i \right ) }},
  \end{aligned}
\end{equation}
\begin{equation}\label{2 reexpression of zetan1n}
  \begin{aligned}
    \hat{\zeta}_{n_1n} \ge \hat{c}_{n_1n}^{\left ( i \right ) } \prod_{m\in \mathcal{K}}\left ( L p_{mn}\theta_{mn_1} \right )^{\hat{a}_{mn_1n}^{\left ( i \right ) }}.
  \end{aligned}
\end{equation}
Subsequently, in the $i+1$-th iteration, problem (\ref{optimization problem 1}) can be reformulated as a standard GP problem as follows
\begin{subequations}\label{optimization problem 1.2}
\begin{align}
\max_{\mathbf{P}^{\left ( i+1 \right )},\pmb{\kappa}^{\left ( i+1 \right )}} \prod_{n=1}^{N} \kappa_n^{w_n^{\left ( i \right ) }}  \ \ \tag{\ref{optimization problem 1.2}}
\end{align}
\begin{alignat}{2}
\text{s.t.}\ & \left ( c_n^{\left ( i \right ) } \right )^2 \prod_{m\in \mathcal{K}}\left ( L p_{mn}\theta_{mn} \right )^{2a_{mn}^{\left ( i \right ) }} \ge \varpi_n \kappa_n,\ \forall n \in \mathcal{N}, \label{problem 1.2 constrited 1} \\
    & \left ( \hat{c}_{n_1n}^{\left ( i \right ) } \right )^2  \prod_{m\in \mathcal{K}}\left ( L p_{mn}\theta_{mn_1} \right )^{2\hat{a}_{mn_1n}^{\left ( i \right ) }} \nonumber \\
    &\ge \hat{\varpi}_{n_1 n} \kappa_n,\  \pi_n=\pi_{n'}, n\le n', \forall n,n_1 \in \mathcal{N}, \label{problem 1.2 constrited 2} \\
    &(\ref{problem 11 constrited 3}), (\ref{orignal constrited 2}), (\ref{orignal constrited 4}),
\end{alignat}
\end{subequations}
where $w_n^{\left ( i \right ) } = \frac{\eta}{\ln 2} \left ( \rho_n^{\left ( i \right ) } -a_n \hat{\rho}_n^{\left ( i \right ) } \right ) $, $\forall n \in \mathcal{N}$.
Although the GP problem is difficult to solve directly, we can convert it into a convex optimization problem by logarithmically transforming the variables. Then, the transformed problem can be effectively solved using interior-point methods to obtain a solution for the problem ($\ref{optimization problem 1.2}$). Standard GP problems can be directly solved using the CVX software toolkit \cite{grant2009cvx}.
The iterative algorithm for solving the power allocation problem with a fixed UE clustering matrix is presented in Algorithm \ref{SPA}, based on the discussion above.

\begin{algorithm}[htbp]
	\caption{SCA Based DL Power Allocation Algorithm (SPA)}\label{SPA}
	Initialize power allocation $\left \{ \mathbf{p}^{(0)}  \right \} $, SINR $\left \{ \pmb{\kappa}^{(0)}\right \}$, the OF of problem (\ref{optimization problem 11}) $\text{Obj}_{1}^{\left ( 0 \right ) }$; iteration number $i = 0$ and the upper bound $T_{out}$ and error tolerance $\xi$;

\While{$\left | \text{Obj}_1^{\left ( i+1 \right ) } - \text{Obj}_1^{\left ( i \right ) }  \right | /\text{Obj}_1^{\left ( i \right ) } \ge \xi $ and $t \le T_{out}$}{

           Solve problem (\ref{optimization problem 1.2}) with GP solver to obtain $\left \{ \mathbf{p}^{(i+1)},\pmb{\kappa}^{(i+1)} \right \} $;

           Update $\left \{ w_n^{(i+1)}, \rho_n^{(i+1)},\hat{\rho}_{n}^{(i+1)},a_{mn}^{(i+1)},c_n^{(i+1)},a_{mn_1n}^{(i+1)},c_{n_1n}^{(i+1)}  \right \}$, ${\forall m,n,n_1}$, and calculate $\text{Obj}_{1}^{\left ( i+1 \right)}$;

           Update $i = 1 + 1$;
        }
\end{algorithm}

Finally, we analyze the convergence of Algorithm \ref{SPA}. In the $i$-th iteration, we obtain the value of $\text{Obj}_{1}^{(i)}$. After optimization in the $i$-th iteration, the approximation of the OF is greater than $\text{Obj}_{1}^{(i)}$. Furthermore, since $\text{Obj}_{1}^{(i+1)}$ is greater than its approximation, we have $\text{Obj}_{1}^{(i+1)} \geq \text{Obj}_{1}^{(i)}$. Besides, it is worth noting that the OF of the problem (\ref{optimization problem 11}) has an upper bound due to the individual minimum rate requirements of URLLC UEs and power constraint at AP. Thus, Algorithm \ref{SPA} is guaranteed to converge.

\subsection{UE Clustering Matrix Design based on Graph Theory}

When the power allocation scheme is fixed, the problem of designing UE clustering matrix can be reformulated as
\begin{subequations}\label{optimization problem 2}
\begin{align}
\max_{\mathbf{X}  }\ \sum_{n\in\mathcal{N}}\hat{R}_n   \ \ \tag{\ref{optimization problem 2}}
\end{align}
\begin{alignat}{2}
\text{s.t.}\  (\ref{orignal constrited 1}), (\ref{orignal constrited 3}).
\end{alignat}
\end{subequations}

Using an exhaustive search to solve the 0-1 UE clustering problem is highly impractical \cite{Dang2022Optimal, Le2021Learning}.
Graph theory is an effective method for solving the UE clustering problem within affordable complexity and has been utilized to address power minimization problems in systems \cite{Guo2019Interference,GuoJoint2022} \footnote{The UE clustering algorithms proposed by these works focus on solving power minimization problems and cannot be used to support URLLC with the goal of maximizing ASR.}.
Inspired by this, we propose a graph theory-based algorithm to efficiently solve the UE clustering problem in NOMA-aided CFmMIMO systems, aiming to support URLLC.
Specifically, we transform the UE clustering problem into the problem of negative loop detection within a weighted directed graph. Based on the OF of the problem (\ref{optimization problem 2}) and the system clustering information, a weighted directed graph is constructed. Two negative loop search algorithms are utilized to identify negative loops within the weighted directed graph for improving the ASR.
The details of the algorithm are provided in the rest of the subsection.
For clarity of expression, we first introduce several definitions \cite{Guo2019Interference,GuoJoint2022}.

\begin{definition}\label{define4}
\textbf{\emph{(Negative loop)}} For a weighted directed graph, if there exists a cycle path starting from a specific vertex, such that traversing this cycle path in a complete loop brings you back to the same vertex, and the sum of the weights of the edges along this cycle path is less than 0, then this cycle can be referred to as a \emph{negative loop}.
\end{definition}

\begin{definition}\label{define1}
\textbf{\emph{(K-shift union)}} For $K$ URLLC UEs belonging to different clusters, denoted as $n_1, n_2, \cdots, n_K$, when the clustering matrix $\mathbf{X}$ can be transformed into $\tilde{\mathbf{X}}$ by $n_1\to n_2, n_2\to n_3,\cdots,n_{K-2}\to n_{K-1}$, and allocating URLLC UE $K-1$ to the cluster where URLLC UE $K$ is allocated, if the ASR satisfies $\sum_{n\in\mathcal{N}}\hat{R}_n(\mathbf{X}) \le \sum_{n\in\mathcal{N}}\hat{R}_n(\tilde{\mathbf{X}})$, then these $K$ URLLC UEs form a \emph{K-shift union}.
\end{definition}

To explain the notation $n\to n'$, assume that with clustering matrix $\mathbf{X}$, URLLC UE $n$ and URLLC UE $n'$ are placed in clusters $g$ and $g'$, respectively, i.e., $x_{gn}=1$ and $x_{g'n'}=1$. Then $n\to n'$ means allocate URLLC UE $n$ to cluster $g'$ and removing URLLC UE $n'$ from cluster $g'$, i.e., $x_{gn}=0$, $x_{gn'}=1$, and $x_{g'n'}=0$.

\begin{definition}\label{define2}
\textbf{\emph{(K-exchange union)}}
For $K$ URLLC UEs belonging to different clusters, denoted as $n_1, n_2, \cdots, n_K$, when the clustering matrix $\mathbf{X}$ can be transformed into $\tilde{\mathbf{X}}$ by $n_1\to n_2, n_2\to n_3,\cdots,n_{K}\to n_{1}$, if the ASR satisfies $\sum_{n\in\mathcal{N}}\hat{R}_n(\mathbf{X}) \le \sum_{n\in\mathcal{N}}\hat{R}_n(\tilde{\mathbf{X}})$, then these $K$ URLLC UEs form a \emph{K-exchange union}.
\end{definition}

Note that the \emph{exchange union} and the \emph{shift union} are different from each other. The former changes the clustering of URLLC UEs without altering the number of URLLC UEs at each cluster, while the latter modifies the number of URLLC UEs at each cluster. Based on these definitions, we provide the definition of the \emph{all-stable solution}.

\begin{definition}\label{define3}
\textbf{\emph{(All-stable solution)}}
For clustering matrix $\mathbf{X}$, if there is no \emph{shift union} or \emph{exchange union} exists with all constraints of problem (\ref{optimization problem 2}) satisfied, then the clustering matrix $\mathbf{X}$ can be called an \emph{all-stable solution}.
\end{definition}

Changing the UE clustering matrix for increasing the ASR can be transformed as a search for either a \emph{shift union} or an \emph{exchange union}.
When unions that meet the constraints of problem (\ref{optimization problem 2}) cannot be found, the problem (\ref{optimization problem 2}) is solved, culminating in the attainment of the optimal UE clustering matrix.

Searching the union among UEs can be achieved by analyzing the variation of $\sum_{n\in\mathcal{N}}\hat{R}_n(\mathbf{X})$ upon changing the UE clustering matrix $\mathbf{X}$. By decomposing the expression of $\sum_{n\in\mathcal{N}}\hat{R}_n(\mathbf{X})$ into clusters, the rate variable of cluster $g$, $\forall g \in \mathcal{G}$ can be defined as
\begin{equation}\label{rate of cluster}
\omega_g (\mathbf{X} ) = \sum_{n\in\mathcal{N}}x_{gn}\left ( \ln(1+\bar{\gamma}_n) - a_n{M(\bar{\gamma}_n) } \right ).
\end{equation}
The interdependence between the ASR and $\omega_g (\mathbf{X} )$ can be expressed as $\sum_{n\in \mathcal{N} }\hat{R}_n(\mathbf{X}) =\frac{\eta}{\ln2}  \sum_{g\in \mathcal{G} } \omega_g (\mathbf{X})$.

A weight directed graph $D\left ( \mathcal{N}_s, \varepsilon ;\mathbf{X} \right )$ is constructed based on (\ref{rate of cluster}), where $\mathcal{N}_s$ represents the set of nodes comprising the URLLC UEs, and $\varepsilon$ denotes the set of edges connecting two URLLC UEs belonging to different clusters.
Denote the adjacency matrix of the graph $D\left ( \mathcal{N}_s, \varepsilon ;\mathbf{X} \right )$ as $\mathbf{Z} = \left [ z_{ij} \right ]_{i,j\in\mathcal{N} }$, where $z_{ij}$ is given by
\begin{equation}\label{element of adjacency matrix}
  \begin{aligned}
  &z_{ij}=\\
   &\left \{
  \begin{aligned}
  &\omega_{\pi_j}(\mathbf{X} )-\omega_{\pi_j}(x_{\pi_ji}=1,x_{\pi_jj}=0,\mathbf{X}_{-i,j} ),& \pi_i \ne \pi_j,  \\
  &\infty, & \pi_i = \pi_j,\\
  \end{aligned}
  \right .
  \end{aligned}
\end{equation}
where $\mathbf{X}_{-i,j}$ represents the clustering matrix $\mathbf{X}$, with URLLC UEs $i$ and $j$ excluded. Therefore, $(x_{\pi_ji}=1,x_{\pi_jj}=0,\mathbf{X}_{-i,j} )$ represents that UE $i$ is allocated to cluster $\pi_j$, URLLC UE $j$ is not allocated to cluster $\pi_j$, and the clustering matrix of URLLC UEs excluding URLLC UEs $i$ and $j$ is consistent with $\mathbf{X}$. Then, the relationship between problem (\ref{optimization problem 2}) and graph $D\left ( \mathcal{N}_s, \varepsilon ;\mathbf{X} \right )$ can be described by the following lemma.
\begin{lemma}\label{Prop of equal}
  For $K$ URLLC UEs belonging to different clusters, denoted as $n_1,n_2,\cdots,n_K$, if these URLLC UEs form a \emph{K-exchange union}, then these URLLC UEs can form a \emph{negative loop} $n_1 \mapsto n_2 \mapsto \dots \mapsto n_K \mapsto n_1$ in the graph $D\left ( \mathcal{N}_s, \varepsilon ;\mathbf{X} \right )$.
\end{lemma}
\begin{proof}
  Denote the UE clustering matrix before and after $n_1 \to n_2\to \dots \to n_K \to n_1$ as $\mathbf{X}$ and $\tilde{\mathbf{X}}$, respectively, then the difference of ASR is given by
  \begin{equation}\label{deltaR}
    \begin{aligned}
    &\Delta \hat{R} =\sum_{n\in \mathcal{N}}\hat{R}_n(\mathbf{X})-\sum_{n\in \mathcal{N}}\hat{R}_n(\tilde {\mathbf{X}})\\
    &= \sum_{n\in \mathcal{N}} \frac{\eta}{\ln2}  \left ( \ln(1+\bar{\gamma}_n\left (\mathbf{X}\right ) ) - a_nM(\bar{\gamma}_n ({\mathbf{X}} ))  \right )\\
    & - \sum_{n\in \mathcal{N}}\frac{\eta}{\ln2}\left ( \ln(1+\bar{\gamma}_n\left (\tilde{\mathbf{X}}\right )) - a_nM(\bar{\gamma}_n (\tilde{\mathbf{X}} ))  \right ).
    \end{aligned}
  \end{equation}

  Assuming that URLLC UEs $n_i$ is allocated to clusters $\pi_i$, $i = 1,2,\cdots,K$. When a URLLC UE enters a new cluster, the ranking of URLLC UEs in this cluster is reevaluated, causing the SINR and ergodic rate of each URLLC UE to be recalculated. However, the SINR of URLLC UEs outside of these $K$ clusters is not affected, and their ergodic rate remainS unchanged by $n_1 \to n_2\to \dots \to n_K \to n_1$. Based on (\ref{deltaR}), the relation between $\Delta \hat{R}$ and $\sum_{k=1}^{K}z_{kj}$  can be expressed as
  \begin{equation}
    \begin{aligned}
    \Delta \hat{R} =& {\textstyle \sum_{k=1}^{K}} {\textstyle \sum_{n=1}^{N}} x_{\pi_kn}\left ( \eta\log_2(1+\bar{\gamma}_n\left (\mathbf{X}\right ))  \right .\\
    &\left .   -a_n\sqrt{{\eta V(\bar{\gamma}_n(\tilde{\mathbf{X}} ))} }\log_2e  \right )\\
    &- {\textstyle \sum_{k=1}^{K}} {\textstyle \sum_{n=1}^{N}}\tilde {x}_{\pi_kn}\left ( \eta\log_2(1+\bar{\gamma}_n\left (\mathbf{X}\right )) \right . \\
    &\left . - a_n\sqrt{{\eta V(\bar{\gamma}_n(\tilde{\mathbf{X}} ))} }\log_2e  \right )\\
    =&\frac{\eta}{\ln 2} {\textstyle \sum_{k=1}^{K}}z_{kj},
    \end{aligned}
  \end{equation}
  where $ j = \mod(k,K)+1 $. Thus, a \emph{K-exchange union} corresponds to a \emph{negative loop} in the graph $D\left ( \mathcal{N}_s, \varepsilon ;\mathbf{X} \right )$.
\end{proof}

Based on Lemma \ref{Prop of equal}, the ASR can increase when finding \emph{exchange union} and updating the clustering matrix.
However, the number of URLLC UEs in each cluster remains constant after the cluster changing process $n_1 \to n_2\to \dots \to n_K \to n_1$. To overcome this limitation, we extend the weight directed graph $D\left ( \mathcal{N}_s, \varepsilon ;\mathbf{X} \right )$ to $D\left ( \mathcal{N}^e, \varepsilon ;\mathbf{X} \right )$ by introducing virtual URLLC UEs to each cluster. Specifically, a virtual URLLC UE, labeled as $n_g^v$, is added to cluster $g$ to allow for arbitrary changes in the number of URLLC UEs at each cluster.
These virtual URLLC UEs are assigned an ergodic data rate of 0 and do not receive any power allocation, nor do they participate in the sorting process. Thus, the presence of virtual URLLC UEs does not affect the ergodic rate of real URLLC UEs. Updating clustering matrix $\mathbf{X}$ to $\tilde{\mathbf{X}}$ by adding URLLC UE $n_{K-1}$ to the cluster of URLLC UE $n_K$ after $n_1\to n_2,\ldots,n_{K-2}\to n_{K-1}$ can be achieved by $n_1\to n_2,\ldots,n_{K-2}\to n_{K-1},n_{K-1}\to n_{\pi_K}^v, n_{\pi_K}^v\to n_1$. As a result, all \emph{shift unions} can be converted into \emph{exchange unions}.
For the convenience of description, the \emph{negative loop} with all URLLC UEs in different clusters is called the \emph{negative differ-cluster loop}.
The theorem about the \emph{all-stable solution} and the \emph{negative differ-cluster loop} is presented in the following.

\begin{theorem}\label{therom1}
  A clustering matrix $\mathbf{X}$ is considered an \emph{all-stable solution} if there is no \emph{negative differ-cluster loop} that can increase the ASR while satisfying all the constraints of the problem (\ref{optimization problem 2}).
\end{theorem}
\begin{proof}
Based on Lemma \ref{Prop of equal} and the fact that all \emph{shift unions} can be converted into \emph{exchange unions}, we can infer that the presence of either a \emph{shift union} or an \emph{exchange union} among all the URLLC UEs, including the virtual ones, necessarily implies the existence of a \emph{negative differ-cluster loop} in the graph $D\left ( \mathcal{N}^e, \varepsilon ;\mathbf{X} \right )$.
Therefore, if no \emph{negative differ-cluster loop} which can increase the OF of the problem (\ref{optimization problem 2}) exists, it follows that there are on \emph{shift union} or \emph{exchange union} which can increase the ASR while satisfying all the constraints.
Consequently, the clustering matrix $\mathbf{X}$ is considered an \emph{all-stable solution}.
\end{proof}

\begin{figure*}[ht]
	\centering
	\includegraphics[scale=0.8]{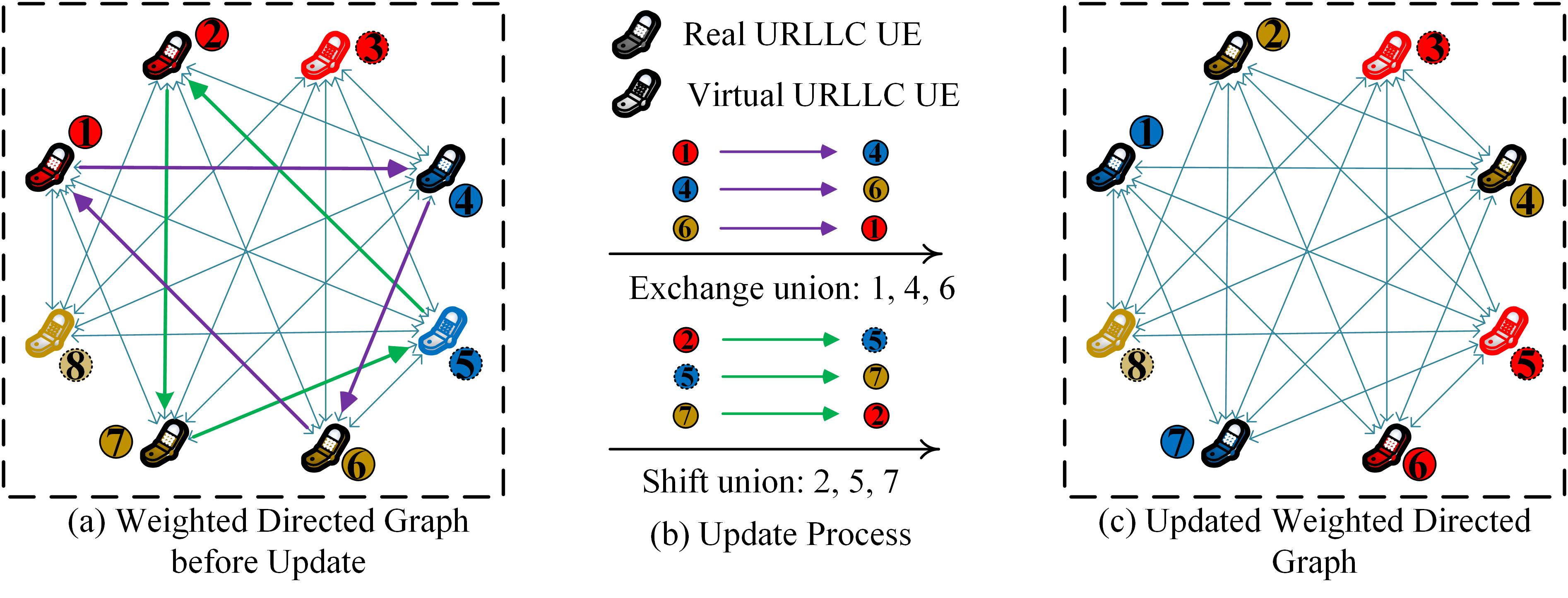}
	\caption{Illustration of the weighted directed graph, \emph{shift league}, and \emph{exchange league} in NOMA-aided CFmMIMO system.}
	\label{fig:digraph}
\end{figure*}

To provide a clearer explanation of the concepts of the weighted directed graph, \emph{shift union}, and \emph{exchange union}, an example of the weighted directed graph within a NOMA-aided CFmMIMO system is illustrated in Fig. \ref{fig:digraph}.
The system comprises five real URLLC UEs,  which are allocated to three clusters distinguished by different colors, as illustrated in Fig. \ref{fig:digraph}(a).
Besides, each cluster also contains a virtual URLLC UE, which can be used to convert \emph{shift unions} to \emph{exchange unions}.
The weighted directed graph within the NOMA-aided CFmMIMO system comprises a set of nodes and edges. These nodes encompass both real URLLC UEs and virtual URLLC UEs. Note that edges in the graph exclusively exist between URLLC UEs allocated to different clusters, denoted by nodes with differing colors as illustrated.
Within this weighted directed graph, two \emph{negative differ-cluster loops}, denoted as \ding{172}$\mapsto$\ding{175}$\mapsto$\ding{177}$\mapsto$\ding{172} and \ding{173}$\mapsto$\ding{176}$\mapsto$\ding{178}$\mapsto$\ding{173}, have been searched. Among them, URLLC UEs \ding{172}, \ding{175}, and \ding{177} form an \emph{exchange union}, as they are all real URLLC UEs. URLLC UEs \ding{173}, \ding{176}, and \ding{178} form a shift union since the URLLC UE \ding{176} is a virtual URLLC UE.
Update the weighted directed graph through \ding{172}$\to$\ding{175}, \ding{175}$\to$\ding{177}, \ding{177}$\to$\ding{172} and \ding{173}$\to$\ding{176}, \ding{176}$\to$\ding{178}, \ding{178}$\to$\ding{173}. URLLC UEs can be allocated to the corresponding cluster, as shown in Fig. \ref{fig:digraph}(c).
Note that the initial cluster scheme is considered an \emph{all-stable solution} if there is no \emph{shift union} or \emph{exchange union} in the system.

%************************************************************************************
Based on Theorem \ref{therom1}, the ASR can be increased by identifying \emph{negative differ-cluster loops} in the graph $D\left ( \mathcal{N}^e, \varepsilon ;\mathbf{X} \right )$ and updating the clustering matrix until no \emph{negative differ-cluster loop} can be searched in the graph $D\left ( \mathcal{N}^e, \varepsilon ;\mathbf{X} \right )$.
To achieve this objective, an algorithm based on graph theory is proposed, which is given in Algorithm \ref{GBUA}.
The UE clustering algorithm iterates continuously until it converges.
In each iteration, we utilize a negative loop detection algorithm to identify \emph{negative differ-cluster loops}. Those loops that do not satisfy UE QoS constraints are placed in the invalid loop set, while the loops that satisfy the constraints are used to update the UE clustering matrix until no such negative loops can be searched in the graph $D\left ( \mathcal{N}^e, \varepsilon ;\mathbf{X} \right )$.
%Note that our proposed UE clustering algorithm is applicable to both finite and infinite blocklength regimes.
\begin{algorithm}[htbp]
	\caption{Graph Theory Based UE Clustering Matrix Design Algorithm}\label{GBUA}
	\KwIn{UE clustering matrix design problem (\ref{optimization problem 2}), UE clustering matrix $\mathbf{X}^{(k)}$, Invalid loop set $\mathcal{S} = \phi$.}
	\KwOut{UE clustering matrix $\mathbf{X}^{(k+1)}$.}
		
	\Repeat{Cannot find an appropriate \emph{negative differ-cluster loop}}{
		Create graph  $D\left ( \mathcal{N}^e, \varepsilon ;\mathbf{X} \right )$;

        Calculate the adjacent matrix $\mathbf{Z}=(z_{ij})_{i,j\in\mathcal{N}^e }$ of graph $D\left ( \mathcal{N}^e, \varepsilon ;\mathbf{X} \right )$ according to (\ref{element of adjacency matrix});

        Search for the \emph{negative differ-cluster loop} in graph $D\left ( \mathcal{N}^e, \varepsilon ;\mathbf{X} \right )$;

        \If{\emph{Negative differ-cluster loop} $\mathcal{L}$ not in $\mathcal{S}$}{

            \If{constraints of problem (\ref{optimization problem 2}) is not satisfied when using $\mathbf{X}^{(k+1)}$}{
                Add $\mathcal{L}$ to $\mathcal{S}$;

                Go bake to step $4$;
            }

            Change clustering matrix $\mathbf{X}^{(k)}$ to $\mathbf{X}^{(k+1)}$ according to loop $\mathcal{L}$;

            Reorder URLLC UEs in $\mathcal{N}$ according to new clustering matrix $\mathbf{X}^{(k+1)}$ by (\ref{rank condition});

            Update $\mathbf{X}^{(k)}=\mathbf{X}^{(k+1)}$;
        }
        Update $\mathcal{S} = \phi$;
	}
\end{algorithm}

In the fourth step of Algorithm \ref{GBUA}, we utilize two algorithms to identify \emph{negative differ-cluster loop}s \cite{Guo2019Interference, GuoJoint2022}. The first algorithm, known as the extended bellman-ford algorithm (EBFA), involves introducing a super node into the graph and connecting it to all nodes in the set $\mathcal{N}_s$. Then, it seeks the shortest path from the super node to all other nodes through relaxation, continuing until no further path can be relaxed. Note that two different URLLC UEs within the same cluster are not present in the same path during relaxation.
The 10th step of EBFA in \cite{Guo2019Interference} requires continuous recursive scheduling of itself. Thus, its computational complexity increases exponentially with the scale of the graph, i.e., the number of URLLC UEs and clusters, which is clearly not conducive to supporting URLLC.
To address this challenge, a second polynomial-time greedy-based suboptimal algorithm (GSA) is used to search for \emph{negative differ-cluster loops} in the graph.
GSA identifies the smallest edge $z_{n_1n_2}$ in $\mathbf{Z}$ and repeats the process until either reaching the maximum number of iteration steps or failing to search for any \emph{negative differ-cluster loops}.
To balance the accuracy and efficiency of GSA, the coefficient $\alpha$ can be utilized to control the number of iterations.

Finally, we analyze the convergence of Algorithm \ref{GBUA}. The number of nodes in the graph is limited by the number of URLLC UEs and clusters, resulting in a finite number of \emph{negative differ-cluster loops} in the graph. Besides, at each iteration of Algorithm \ref{GBUA}, the OF of problem (\ref{optimization problem 2}) increases, while ensuring the constraints of problem (\ref{optimization problem 2}) are satisfied. Consequently, Algorithm \ref{GBUA} can terminate after a finite number of iterations and the convergence of Algorithm \ref{GBUA} is guaranteed.

\subsection{Globe ASR Maximization}
The proposed algorithm for maximizing the global ASR, which integrates power allocation and UE clustering matrix design, is given in Algorithm \ref{alg:5}. As discussed earlier in this section, the algorithm iteratively and alternately solves for $\mathbf{p}$ and $\mathbf{X}$ until a stable optimal ASR is attained.

\begin{algorithm}[htbp]
    \caption{Two-Step Alternating Optimization based Global ASR Maximization Algorithm}\label{alg:5}
    \KwIn{$\mathbf{p}$, $\mathbf{X}$, $M$, $N$, $G$, $L$, $\beta$, error tolerance $\xi$, iteration number $t=1$ and the upper bound $T_{out}$; a feasible solution $\mathbf{p}^{(0)}$ and $\mathbf{X}^{(0)}$ of problem (\ref{orignal optimization problem});}
	\KwOut{The optimal $\mathbf{p}^{(t)}$ and $\mathbf{X}^{(t)}$;}

    \Repeat{$\left | Obj^{(t)} - Obj^{(t-1)} \right | \le  \xi $ or $t \ge T_{out}$}{
        Obtain $\mathbf{p}^{(t)}$ with fixed $\mathbf{X}^{(t-1)}$ based on Algorithm \ref{SPA};

        Update $\mathbf{X}^{(t)}$ with fixed $\mathbf{p}^{(t)}$ based on Algorithm \ref{GBUA};

        Update $t=t+1$;
    }
\end{algorithm}

\textbf{Computational Complexity Analysis:} The SDPT3 optimizer of CVX is employed to solve GP problem (\ref{optimization problem 1.2}) with $(M+1)N$ optimization variables, resulting in a complexity of $O\left(\left ( (M+1)N \right )^3\right)$. Assuming that steps $2$-$6$ in Algorithm \ref{SPA} iterate $\tau_1$ times, the computational complexity of SPA is $O\left(\tau_1  \left ( (M+1)N \right )^3\right)$.
The computational complexity of GSA is $O\left(G (G+N)^2\right)$ \cite{Guo2019Interference}.
Assume GSA iterates $\tau_2$ times, thus the computational complexity of Algorithm \ref{GBUA} is $O\left(\tau_2 G (G+N)^2\right)$.
If Algorithm \ref{alg:5} iterates $\tau$ times, its computational complexity can be approximated as $O\left(\tau \max \left(\tau_1  \left ( (M+1)N \right )^3 , \tau_2  G  (G+N)^2\right) \right)$.

\section{Numerical Results and Discussion}
In this section, extensive numerical results are provided to validate the effectiveness of our proposed algorithms. We begin by verifying the convergence and the complexity of our proposed algorithms. Subsequently, Monte Carlo simulations are provided to demonstrate the close approximation between the LB of the ASR and the actual. Finally, we compare the performance of the proposed algorithms with benchmark algorithms across various scenarios.

\subsection{Simulation Setup and Comparison Algorithms}
In our simulations, we consider a randomly distributed system of APs and URLLC UEs within a rectangular area of $1km \times 1km$. The large-scale fading coefficient, which is influenced by path loss and shadowing effects, is denoted as $\beta_{mn} = PL_{mn} + z_{mn}$, where $PL_{mn}$ represents the path loss component and $z_{mn} \sim \mathcal{CN}(0,\delta_{sh}^2)$ represents the shadowing component following a complex Gaussian distribution with zero mean and variance $\delta_{sh}^2$. We apply a three-slope model proposed in \cite{Ngo2017Cell} to characterize the path loss and use the same parameter settings as in \cite{Ngo2017Cell}. The length of the pilot signal is set to $G$.
To provide a more intuitive representation of our proposed algorithms, we refer to them as S-EBFA and S-GSA, respectively. Both algorithms are based on SCA to solve the power allocation problem. However, when it comes to the UE clustering problem, S-EBFA employs EBFA, while S-GSA utilizes GSA for \emph{negative differ-cluster loop} detection.

{\color{blue}
Unless otherwise specified, the following parameter settings are adopted in the simulations. The number of APs $M$ is set to 120, and the number of URLLC UEs $N$ is set to 40. \footnote{{\color{blue}To ensure that the NOMA-aided CFmMIMO system has sufficient spatial degrees of freedom to support URLLC \cite{Ngo2017Cell, 9650567}, we consider deploying far more APs than the number of URLLC UEs to jointly provide services \cite{ZhangSWIPT2023, Dang2022Optimal, Le2021Learning}.}}
}
The number of antennas each AP $L$ and clusters $G$ is set as $12$ and $N/2$, respectively.
The system bandwidth is set as $10$ MHz, and the noise power spectral density is $-174$ dBm/Hz.
The channel coherent length $\tau_c = 200$, and the decoding error probability is set to $\epsilon = 10^{-6}$. The minimum transmission rate is restricted to $1$ Mbps, while the maximum DL transmission power is set to $23$ dBm. The pilot power is maintained at $20$ dBm, and the SIC coefficient is $c=0.5$.

To validate the effectiveness of our proposed algorithms, we compare their performance against the following benchmark algorithms:

\begin{itemize}	
	\item  \textbf{Gale-Shapley:} In Gale-Shapley algorithm\cite{Xu2018Joint}, each URLLC UE prefers the cluster which has less intra-cluster interference and each cluster prefers the URLLC UE that has the smaller large-scale fading coefficient. Note that the number of URLLC UEs in each cluster does not exceed $\left \lceil \frac{N}{G}  \right \rceil $.

	\item \textbf{BRPA:} Basic random URLLC UE clustering with power allocation optimization algorithm (BRPA) randomly allocates URLLC UEs to different clusters and only optimizes power allocation.
\end{itemize}

\subsection{Convergence and Complexity Analysis}\label{convergence}
\begin{figure*}
	\centering
	\subfloat[]{
	\includegraphics[width=0.32\textwidth]{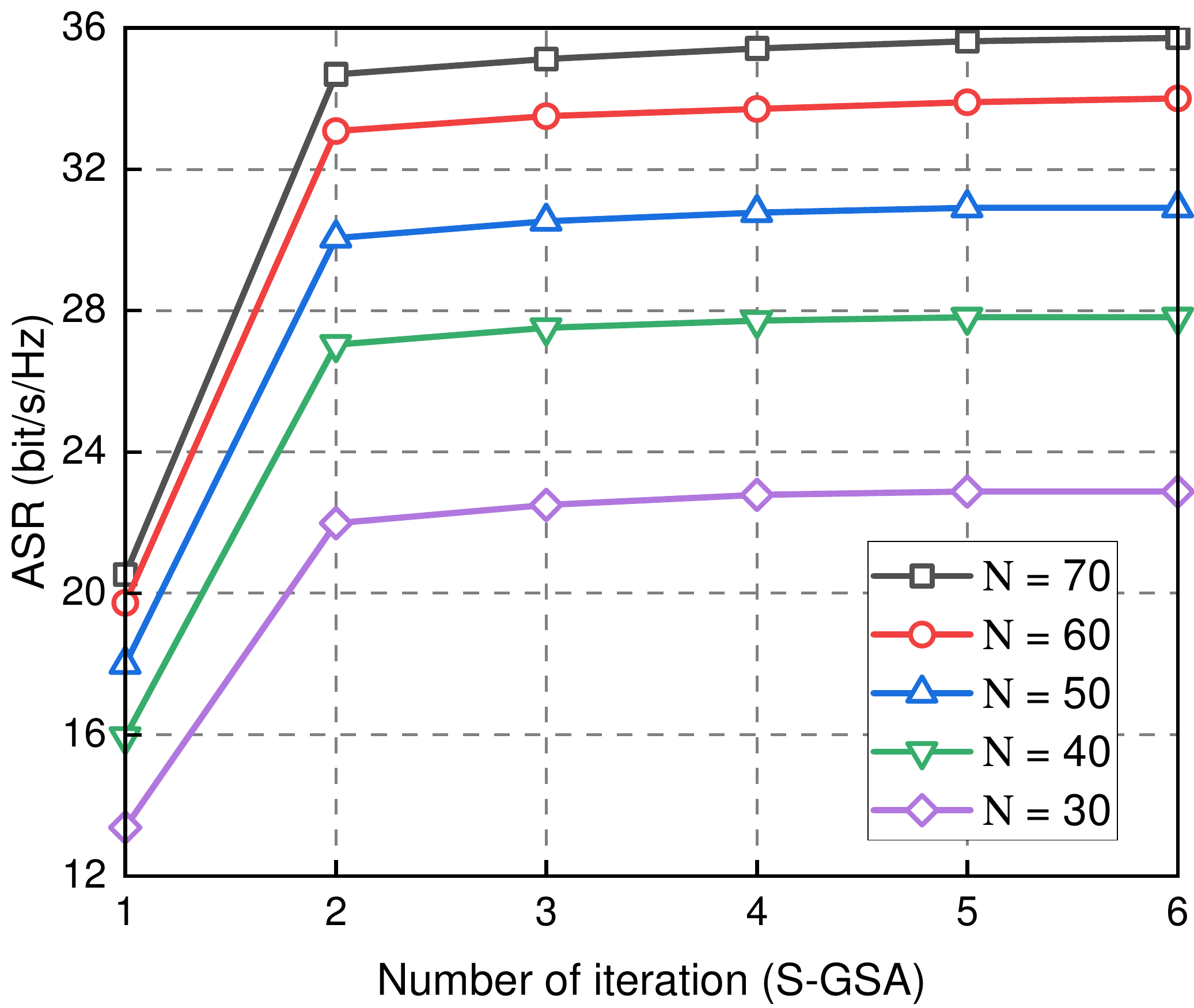}}
    \subfloat[]{
	\includegraphics[width=0.32\textwidth]{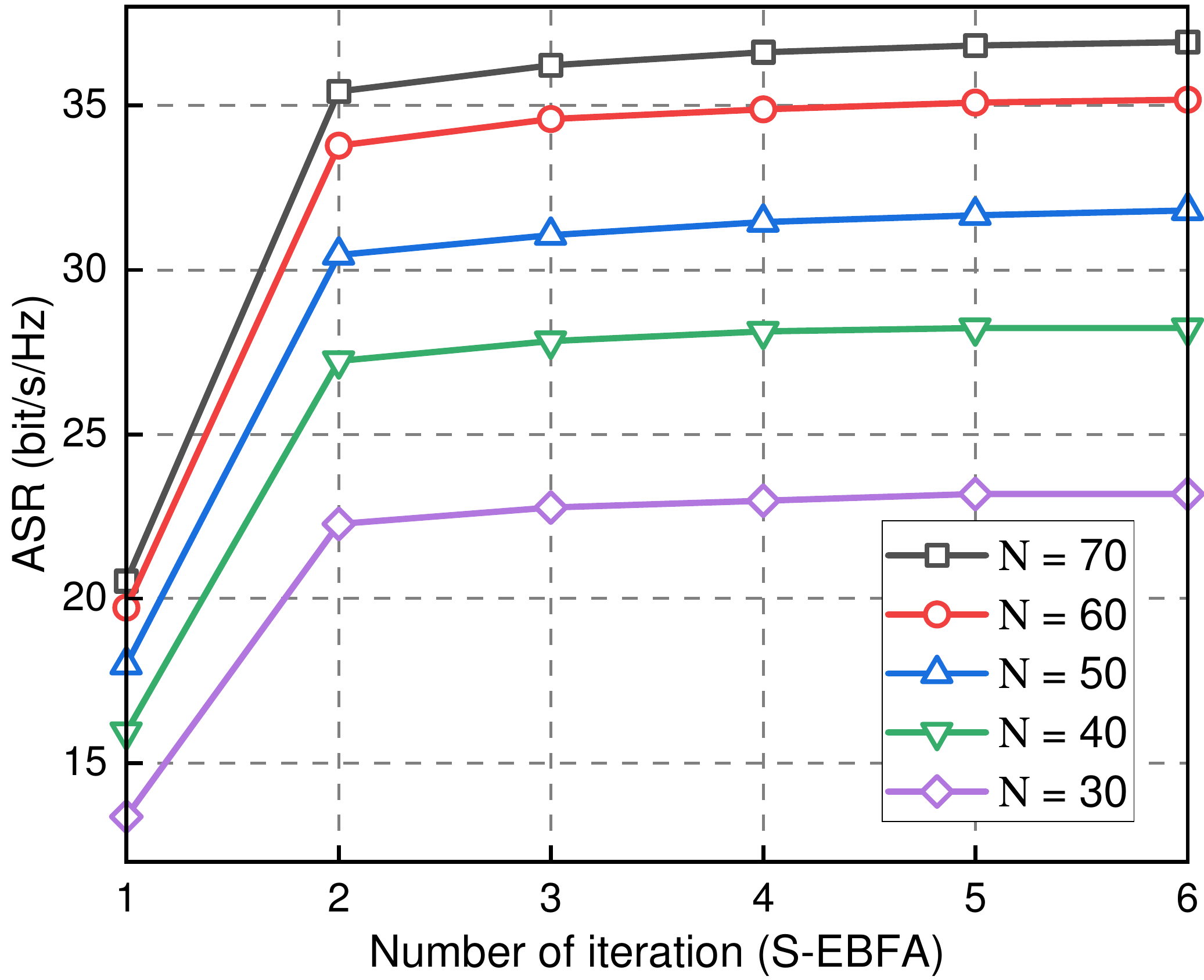}}
	\subfloat[]{
	\includegraphics[width=0.32\textwidth]{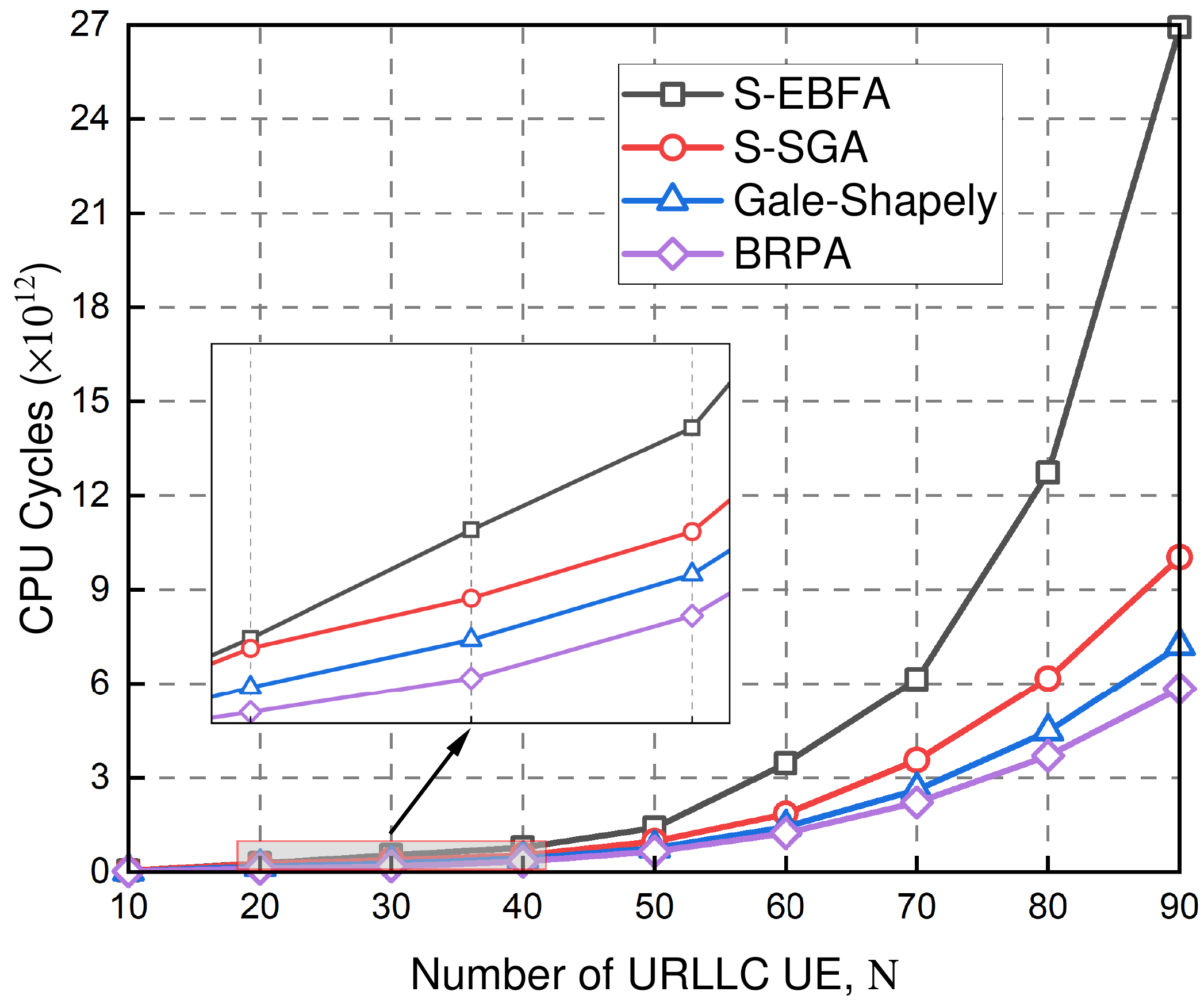}}
    \caption{Convergence and complexity analysis of the proposed algorithm: a. Convergence of S-GSA vs. number of URLLC UEs; b. Convergence of S-EBFA vs. number of URLLC UEs; c. The performance of algorithm complexity vs. the number of URLLC UEs.}
	\label{conCom}
\end{figure*}

We investigate the convergence of S-EBFA and S-GSA under varying numbers of URLLC UEs, as illustrated in Fig. \ref{conCom} (a) and \ref{conCom} (b), respectively.
%Note that the parameter $\alpha$ for S-GSA is set to 10.
The results indicate that both proposed algorithms can converge after 5-6 iterations in various scenarios.

In order to explore the complexity performance of the algorithm, we compared the CPU cycles of the algorithms under different number of URLLC UEs, as shown in Fig. \ref{conCom} (c). CPU cycle is defined by algorithm runtime and CPU operation frequency, which can be used as an indicator of algorithm complexity. It can be seen that the complexity of our proposed algorithm increases with the increase in the number of URLLC UEs. When the number of URLLC UEs is small, the complexity of the S-EBFA algorithm is not significantly different from other algorithms. However, with the increase of URLLC UE, the complexity of the S-EBFA algorithm rapidly increases, which is caused by the exponential complexity of the S-EBFA algorithm. Compared to the S-EBFA algorithm, the complexity of the S-GSA algorithm increases significantly lower with the number of URLLC UEs. The numerical results demonstrate that the S-GSA algorithm with polynomial complexity has a more significant advantage in algorithm complexity compared to the S-EBFA algorithm with exponential complexity, as the number of URLLC UEs increases. In addition, the algorithmic complexity of the Gale Shapely algorithm, Hungarian algorithm, and BRPA algorithm are not significantly different in various cases. Compared to the benchmark algorithms, the S-GSA algorithm has a higher complexity, but due to its polynomial complexity, its complexity does not grow too fast like the S-EBFA algorithm.

\subsection{Performance Analysis}
\begin{figure}[ht]
	\centering
    \setlength{\abovecaptionskip}{-0.1cm}   %调整图片标题与图距离
	\includegraphics[scale=0.33]{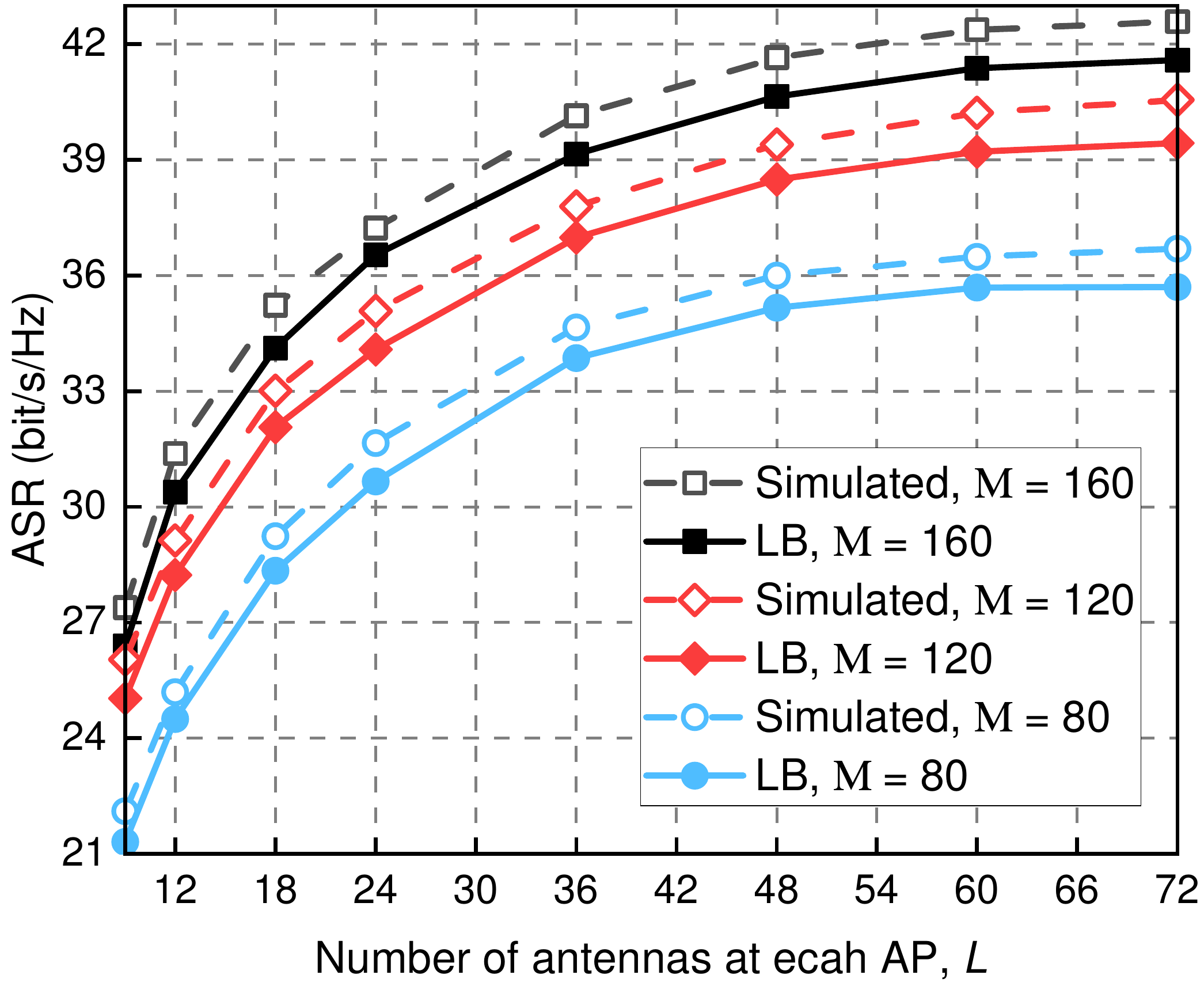}
	\caption{Tightness of derived data rate LB vs. the number of APs.}
	\label{tightnessOfLBimag}
\end{figure}

\begin{figure}[ht]
	\centering
    \setlength{\abovecaptionskip}{-0.1cm}   %调整图片标题与图距离
	\includegraphics[scale=0.33]{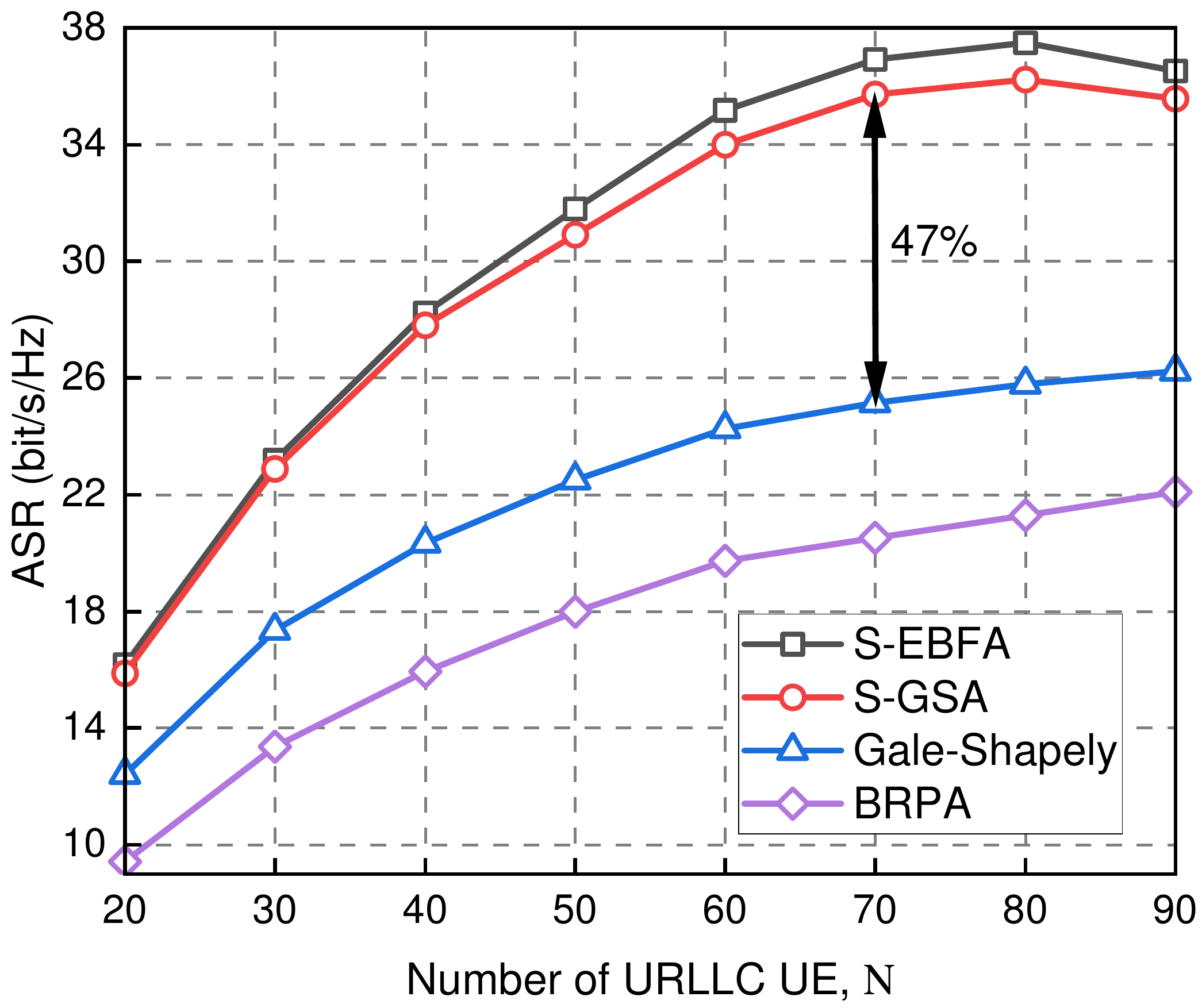}
    \caption{ASR vs. the number of URLLC UEs.}
    \label{differ_users}
\end{figure}

\begin{figure}[ht]
	\centering
    \setlength{\abovecaptionskip}{-0.1cm}   %调整图片标题与图距离
	\includegraphics[scale=0.33]{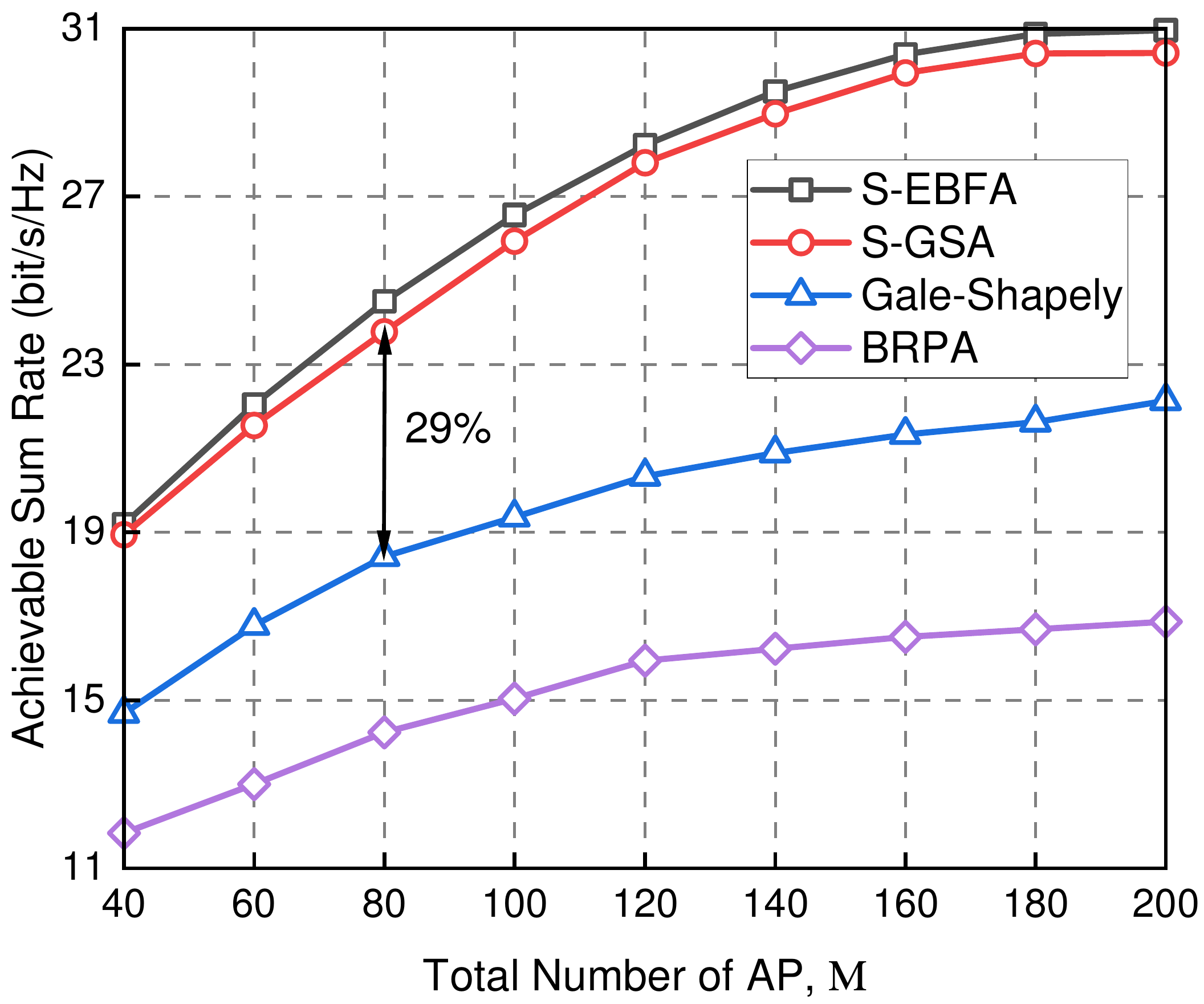}
    \caption{ASR vs. the number of APs.}
    \label{differ_APs}
\end{figure}

\begin{figure*}
	\centering
	\subfloat[]{
	\includegraphics[width=0.32\textwidth]{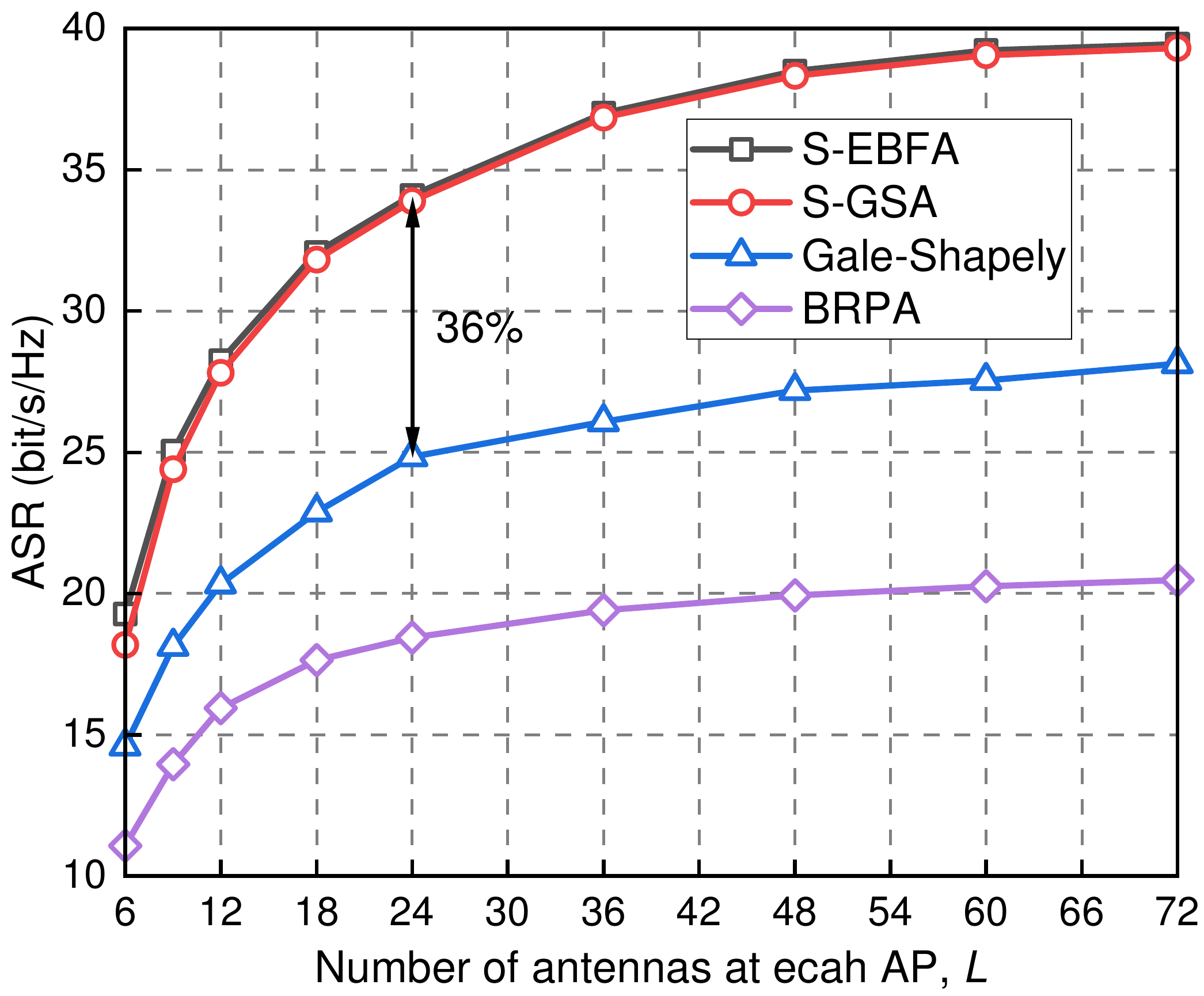}}
    \subfloat[]{
	\includegraphics[width=0.32\textwidth]{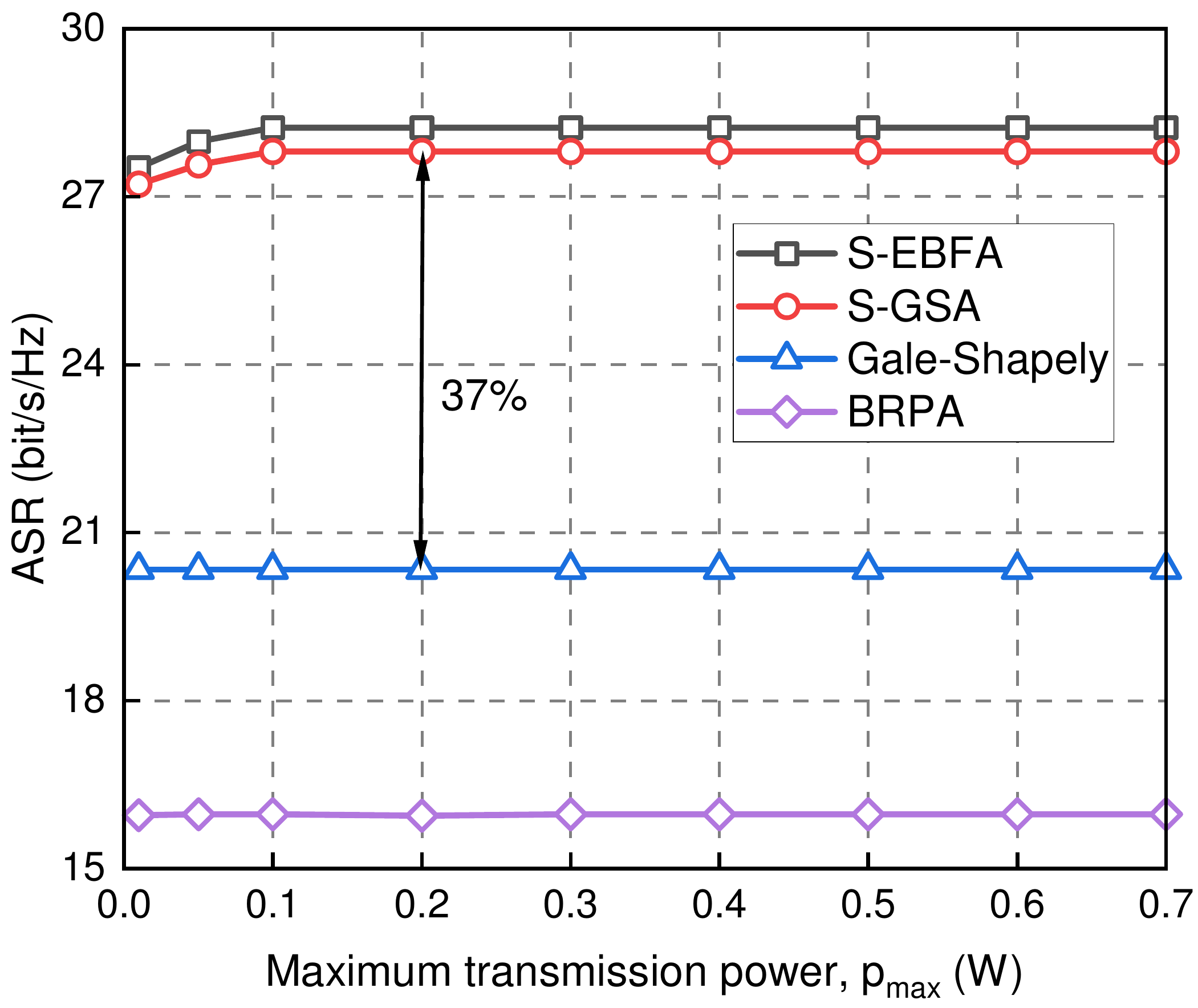}}
	\subfloat[]{
	\includegraphics[width=0.32\textwidth]{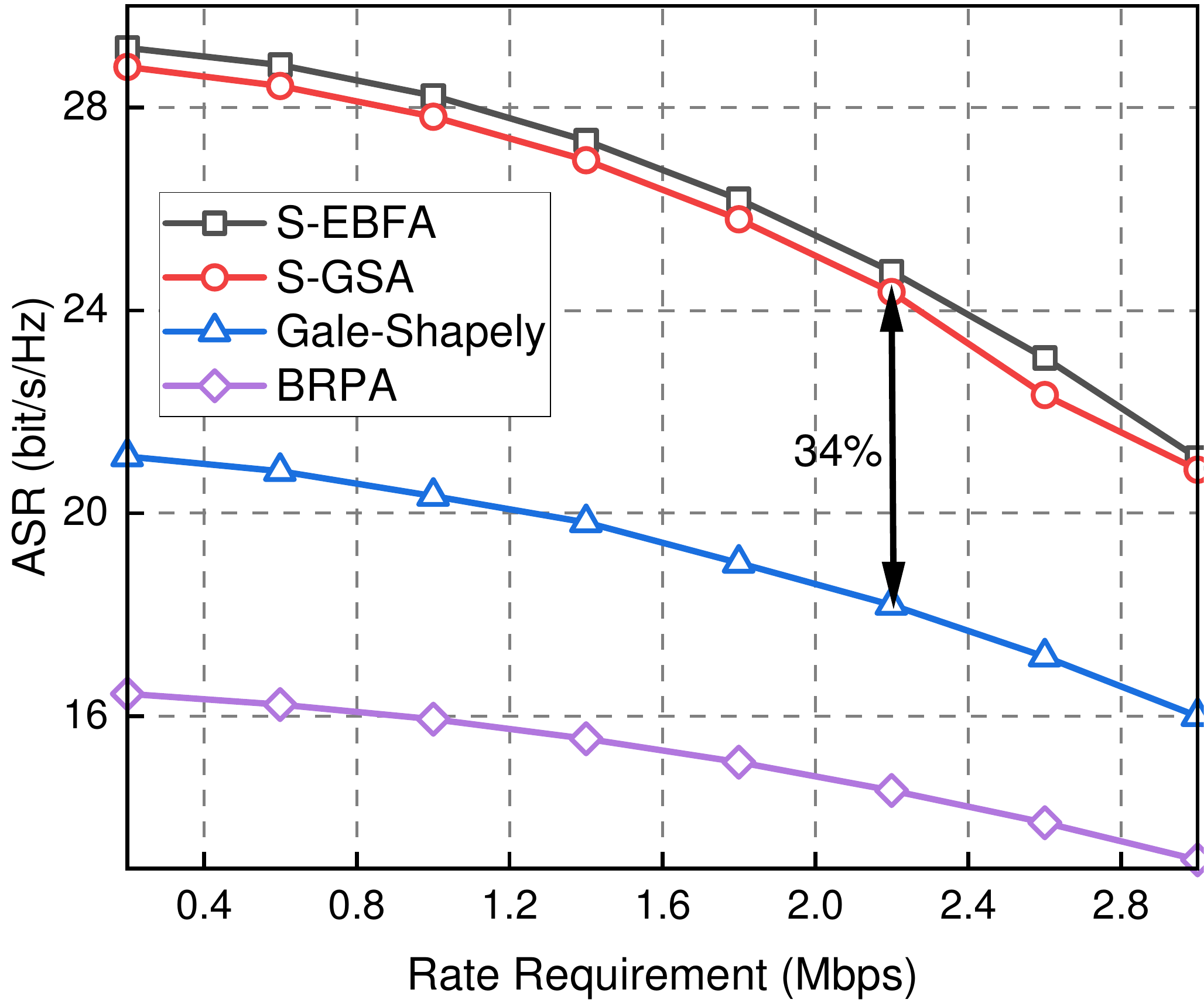}}
    \caption{Performance comparison of different algorithms under different scenarios: a. ASR vs. the number of antennas at each AP; b. ASR vs. maximum transmit power; c. ASR vs. rate requirement.}
	\label{differ}
\end{figure*}

To evaluate the accuracy of the LB derived in Theorem \ref{theorem of LB}, we conducted Monte Carlo simulations by generating random channels one million times and averaging the results. Fig. \ref{tightnessOfLBimag} illustrates that the derived LB exhibits high precision across various different numbers of APs and antennas.

% Users
We evaluate the ASR of both our proposed algorithms and benchmark algorithms under different numbers of URLLC UEs, as illustrated in Fig. \ref{differ_users}. Specifically, we consider a range of URLLC UEs from $20$ to $90$, while keeping the experimental parameter settings consistent with those outlined in Section \ref{convergence}.
We observe an initial increase in the ASR as the number of URLLC UEs increases, reaching its peak at $N = 80$. This trend can be attributed to the increase in the number of URLLC UEs. However, the ASR subsequently decreases as the number of URLLC UEs increases, which can be attributed to the increase in inter-cluster interference and pilot length with the rise of URLLC UEs.
Furthermore, the proposed algorithms outperform the benchmark algorithms across various numbers of URLLC UEs, which is attributed to the joint optimization of power allocation and UE clustering.
Specifically, at $N = 70$, algorithm S-GSA surpasses the performance of the BRPA by $73\%$ and the Gale-Shapley by $47\%$.

% APs
We evaluate the performance of the proposed algorithms and other benchmark algorithms under varying numbers of APs, as illustrated in Fig. \ref{differ_APs}.
It is evident that the ASR rises as the number of APs increases.
This can be attributed to the fact that the inter-cluster interference increases linearly with the number of APs, whereas the desired signal and intra-cluster interference increase with the square of the number of APs, resulting in a relatively weaker influence of inter-cluster interference as the number of APs grows.
When the number of AP is large, the ASR stabilizes as it is primarily determined by the desired signal and intra-cluster interference, while the inter-cluster interference becomes negligible in comparison. When the number of APs is $80$, our proposed S-GSA outperforms benchmark algorithms such as BRPA and Gale-Shapley by $67\%$ and $29\%$ respectively.

% Antennas
We evaluate the ASR of the proposed algorithms and compare them with benchmark algorithms for different numbers of antennas at each AP, as illustrated in Fig. \ref{differ} (a). The ASR increases with the number of antennas at the AP and it can be seen that as the number of antennas at each AP increases, the growth rate of ASR gradually slows down. When deploying the benchmark algorithms in the NOMA-aided CFmMIMO system for supporting URLLC, ASR stabilizes as the number of antennas at each AP increases, proving that the ergodic rate has an upper bound as the number of antennas per AP goes to infinity.
When the number of antennas is $L=24$, the proposed S-GSA achieves $84\%$ and $36\%$ higher performance compared to the BRPA and Gale-Shapley, respectively.

% max transmit power
Fig. \ref{differ} (b) shows that the influence of transmit power on the ASR.
The increase in transmission power indirectly mitigates the effects of noise on the DL data transmission. However, excessive transmission power does not yield significant improvements in achievable rates. This suggests that a small and optimal DL power is adequate for minimizing the impact of noise. When the maximum transmission power is set to $p_d = 0.2\ \mathrm{W}$, S-GSA achieves a $74\%$ improvement over BRPA and a $37\%$ improvement over Gale-Shapley.

% R_req
We evaluate the performance of the proposed algorithms and the benchmark algorithms under varying minimum rate requirements, as illustrated in Fig. \ref{differ} (c). With an increase in the minimum rate requirement, the ASR gradually decreases. The improvement of the proposed algorithms diminishes as the minimum rate requirement increases, and its final performance closely aligns with that of the benchmark algorithms. Moreover, the number of available schemes satisfying the minimum rate requirement decreases as the requirement becomes more stringent, resulting in a decrease in the ASR. Therefore, the performance of the proposed algorithms and the benchmark algorithms is similar. When minimum rate requirement is $2.2$ Mbps, the ASR of algorithm S-GSA demonstrates a $68\%$ improvement over BRPA and an $34\%$ improvement over Gale-Shapley.

\section{Conclusion}
In this paper, we have investigated the performance of the DL NOMA-aided CFmMIMO system for URLLC in terms of ASR, with consideration of FBC.
Specifically, we have derived the LB for the ergodic data rate, while taking into account the effects of inter-cluster pilot contamination, inter-cluster interference, and imperfect SIC. To maximize the ASR while ensuring the minimum data rate constraint, we proposed a joint optimization framework for power allocation and UE clustering.
To solve the problem efficiently, we propose a two-step iterative algorithm where SCA is exploited to transform power allocation into a series of GPs, and UE clustering is converted into finding \emph{differ-cluster negative loop} based on graph theory. We have conducted simulations, to validate our analytical and optimization work. The results have shown that the derived LB was tight and the proposed algorithms achieved significant gains in terms of ASR compared with the benchmark algorithms under various scenarios.

Our study on supporting URLLC in the NOMA-aided CFmMIMO system remains at a preliminary theoretical level and requires further investigation. The performance of the NOMA-aided CFmMIMO system in supporting URLLC is significantly influenced by precoding schemes. Therefore, designing a precoding scheme that satisfies the URLLC requirements poses considerable challenges.
Besides, the presence of a large number of UEs in a 6G network can impose a significant burden on the backhauls. Therefore, it is necessary to design a distributed NOMA-aided CFmMIMO system where multiple CPUs are deployed to ensure low load in the backhauls to better support URLLC.
\begin{appendices}
\section{Proof of Theorem 1}\label{appendix A}
\begin{figure*}
  \begin{equation}\label{appendix111}
   \begin{aligned}
      \mathbb{E}\left \{ \left | Y_{rici,n} \right |^2 \right \}&=\mathbb{E}\left \{ \left (  {\textstyle \sum_{m=1}^{M}}{\textstyle \sum_{n'=n+1}^{N}} x_{\pi_nn'} \sqrt{p_{mn'}} \left (\iota_{mnn'}q_{n'} - \mathbb{E}\left \{ \iota_{mnn'} \right \}\hat{q}_{n'} \right ) \right )^2 \right \} \\
           & = {\textstyle \sum_{n'=n+1}^{N}}L\left ( 2-2c_{n'} \right ) x_{\pi_nn'}  \left ({\textstyle \sum_{m=1}^{M}} \sqrt{p_{mn'}\theta_{mn}}  \right )^2+ {\textstyle \sum_{n'=n+1}^{N}} x_{\pi_nn'} {\textstyle \sum_{m=1}^{M}}p_{mn'}\beta_{mn}  .
   \end{aligned}
  \end{equation}
  {\noindent} \rule[-10pt]{18cm}{0.05em}
\end{figure*}
Based on $\gamma_n = \min \left ( \gamma_n^{n},\gamma_{n_1}^n \right )$ and $\bar{\gamma}_n = \mathbb{E}_{\gamma_n}^{-1}\left \{ \gamma_n^{-1} \right \}$, we have $\bar{\gamma}_n = \min(\bar{\gamma}_n^n,\bar{\gamma}_{n_1}^n)$, where $\bar{\gamma}_n^n = \mathbb{E}_{\gamma_n^n}^{-1}\left \{ \left ( \gamma_n^{n} \right )^{-1}  \right \}$ and $\bar{\gamma}_{n_1}^n = \mathbb{E}_{\gamma_{n_1}^n}^{-1}\left \{ \left ( \gamma_{n_1}^{n} \right )^{-1}  \right \}$. Then we first derive the expression of $\bar{\gamma}_n^n$.
According to the independence between estimated channel $\hat{\mathbf{h}}_{mn}$ and estimate error $\mathbf{h}_{mn} - \hat{\mathbf{h}}_{mn}$, we first derive the expression of $\left | Y_{ds,n} \right |^2$
\begin{equation}
  \begin{aligned}
    \left | Y_{ds,n} \right |^2 &= \left (   {\textstyle\sum_{m\in \mathcal{M}}}  \sqrt{p_{mn}}  \mathbb{E}\left \{  \iota_{mnn}  \right \} \right )^2 \\
                                &= L\left ( {\textstyle\sum_{m\in \mathcal{M}}}  \sqrt{p_{mn}\theta_{mn}}  \right )^2.
  \end{aligned}
\end{equation}
The expression of $\mathbb{E}\left \{ \left | Y_{bu,n} \right |^2  \right \}$ is given as
\begin{equation}
 \begin{aligned}
  &\mathbb{E}\left \{ \left | Y_{bu,n} \right |^2  \right \} \\
  &= \mathbb{E}\left \{ \left (  {\textstyle \sum_{m \in \mathcal{M} }}  \sqrt{p_{mn}}\left (\iota_{mnn} -\mathbb{E}\left \{\iota_{mnn}  \right \} \right )q_{n} \right )^2 \right \}\\
  &= {\textstyle \sum_{m \in \mathcal{M} }}p_{mn}\beta_{mn}.
 \end{aligned}
\end{equation}
$\mathbb{E}\left \{ \left | Y_{ui,n} \right |^2  \right \}$ can be calculated as
\begin{equation}
 \begin{aligned}
  \mathbb{E}\left \{ \left | Y_{ui,n} \right |^2  \right \}= \sum_{g'\ne \pi_n} \sum_{n'\in \mathcal{N}} x_{g'n'}  \sum_{m\in \mathcal{M}}p_{mn'}\beta_{mn}.
 \end{aligned}
\end{equation}
Then we derive the expression of $\mathbb{E}\left \{ \left | Y_{ici,n} \right |^2 \right \}$.
\begin{equation}
 \begin{aligned}
  &\mathbb{E}\left \{ \left | Y_{ici,n} \right |^2 \right \}   =\mathbb{E} \left \{ \left ( \sum_{m=1}^{M} \sum_{n'=0}^{n-1} x_{\pi_nn'} \sqrt{p_{n'}}\iota_{mnn'}q_{n'}  \right )^2 \right\}\\
  &= \sum_{n'=0}^{n-1} x_{\pi_nn'} \left ( \sum_{m=1}^{M}p_{mn'}\beta_{mn} + L\left (\sum_{m=1}^{M} \sqrt{p_{mn'}\theta_{mn}}  \right )^2 \right ).
 \end{aligned}
\end{equation}
The expression of $\mathbb{E}\left \{ \left | Y_{rici,n} \right |^2 \right \}$ is given in (\ref{appendix111}).
By substituting the above formulations into the expectation of (\ref{SINR of MRT}), we obtain the expression of $\hat{\gamma}_n^n$ in (\ref{bar gamma}). Then, we can obtain the expression of $\hat{\gamma}_n^{n_1}$ in a similar way.
\end{appendices}

\bibliographystyle{IEEEtran}
\bibliography{IEEEabrv,v1ref}
\end{document}